\documentclass{article}

\usepackage{arxiv}

\usepackage[utf8]{inputenc} 
\usepackage[T1]{fontenc}    
\usepackage{hyperref}       
\usepackage{nicefrac}       
\usepackage{microtype}      

\title{Efficient Approximation of the Matching Distance for $2$-parameter persistence}
\author{Michael Kerber\\
Graz University of Technology \\
Institut für Geometrie \\
Kopernikusgasse 24, 8010 Graz, Austria \\
\And
Arnur Nigmetov\\
Graz University of Technology \\
Institut für Geometrie \\
Kopernikusgasse 24, 8010 Graz, Austria \\
}

\usepackage{amsmath}%
\usepackage{amssymb}%
\usepackage{amsfonts}
\usepackage{amsthm}

\newtheorem{thm}{Theorem}
\newtheorem{mylem}[thm]{Lemma}

\usepackage{booktabs}
\usepackage{url}
\usepackage{xspace}

\newcommand{\ra}[1]{\renewcommand{\arraystretch}{#1}}
\usepackage{pgf}
\usepackage{graphicx}
\usepackage{pgfplots}
\pgfplotsset{compat=1.14}
\usepackage{pgfplotstable}
\usepackage{epsfig}
\usepackage{color}

\renewcommand{\leq}{\leqslant}
\renewcommand{\geq}{\geqslant}
\renewcommand{\epsilon}{\varepsilon}
\renewcommand{\phi}{\varphi}
\newcommand{\eps}{\varepsilon}

\mathchardef\mhyphen="2D

\newcommand\restr[2]{{
  \left.\kern-\nulldelimiterspace 
  #1 
  \vphantom{\big|} 
  \right|_{#2} 
  }}

\newcommand{\RR}{\mathbb{R}}

\newcommand{\ignore}[1]{}

\newcommand{\libraryname}[1]{{\textsc{#1}}\xspace}

\newcommand{\hera}{\libraryname{Hera}}
\newcommand{\rivet}{\libraryname{RIVET}}

\newcommand{\pleq}{\preceq}

\DeclareMathOperator{\push}{{push}}

\newcommand{\weight}{w}

\newcommand{\dsGH}{\textrm{GH}\xspace}
\newcommand{\dsED}{\textrm{ED}\xspace}
\newcommand{\dsR}{\textrm{RND}\xspace}

\usepackage[capitalize,noabbrev]{cleveref}
\crefname{mydef}{definition}{definitions}
\crefname{mylem}{lemma}{lemmas}
\crefname{mythm}{theorem}{theorems}

\newcommand{\R}{\mathbb{R}}
\newcommand{\po}{\leq}
\newcommand{\slice}{L}
\newcommand{\slices}{\mathcal{L}}
\newcommand{\slicebox}{B}
\newcommand{\sangle}{\gamma}
\newcommand{\sslope}{\lambda}
\newcommand{\smu}{\mu}

\newcommand{\vw}{\vec{w}}

\newcommand{\centerslice}{L_c}
\newcommand{\centersslope}{\sslope_c}
\newcommand{\centersmu}{\smu_c}

\newcommand{\ffunc}{\phi}
\newcommand{\origin}{\mathcal{O}}

\DeclareMathOperator{\weightedpush}{{wpush}}
\DeclareMathOperator{\wpush}{\weightedpush}

\DeclareMathOperator{\dmatch}{d_{M}}

\newcommand\colpt[2]{%
    \left( \begin{array}{c}#1\\#2\end{array} \right)}

\begin{document}

\maketitle

\begin{abstract}
The matching distance is a computationally tractable 
topological measure to compare multi-filtered simplicial complexes.
We design efficient algorithms for approximating the matching distance
of two bi-filtered complexes to any desired precision $\eps>0$.
Our approach is based on a quad-tree refinement strategy introduced
by Biasotti et al.~\cite{italian}, but we recast their approach entirely
in geometric terms. This point of view leads to several novel
observations resulting in a practically faster algorithm.
We demonstrate this speed-up by experimental comparison and provide
our code in a public repository which provides the first efficient publicly
available implementation of the matching distance.
\end{abstract}

\section{Introduction}

Persistent homology~\cite{elz-topological,carlsson-survey,eh-computational,oudot-book} is one of the major concepts in the quickly evolving
field of topological data analysis. The concept is based on the idea
that studying the topological properties of a data set across various scales
yields valuable information that is more robust to
noise than restricting to a fixed scale.

We distinguish the case of \emph{single-parameter persistence},
where the scale is expressed by a single real parameter, and the case of
\emph{multi-parameter persistence}, in which the scale consists
of two or more parameters that vary independently. 
The former case is the predominant one in the literature.
The entire homological multi-scale evolution of the data set can be expressed
by a multi-set of points in the plane, the so-called \emph{persistence diagram}. Moreover, the \emph{interleaving distance} yields a distance measure
between two data sets by measuring the difference in their topological evolution. For a single parameter, this distance can be rephrased as a combinatorial
matching problem of the corresponding persistence diagrams (known as
the \emph{bottleneck distance}) and computed efficiently~\cite{kmn-geometry}.
These results are part of a rich theory of single-parameter persistence,
with many algorithmic results and applications.

The case of \emph{multi-parameter persistence} received significantly less
attention until recently. One reason is the early result that 
a complete combinatorial structure 
such as the persistence diagram does not exist for two or more parameters~\cite{cz-theory}.
Moreover, while the interleaving distance can be straight-forwardly
generalized to several parameters, its computation becomes NP-hard
already for two parameters, as well as any approximation to a factor less than $3$~\cite{bbk-computing}. On the other hand, data sets with several scale
parameters appear naturally in applications, and an efficiently computable
distance measure is therefore highly important.

We focus on the \emph{matching distance}~\cite{cerri2013betti,italian,klo-exact} 
as a computationally tractable lower bound on the interleaving distance~\cite{landi2018rank}. 
It is based on the observation that when restricting the multi-parameter
space $\R^d$ to a one-dimensional affine subspace (that is, a line in $\R^d$),
we are back in the case of single-parameter
persistence. We can compute the bottleneck distance between the two persistence
diagrams restricted to the same subspace. The matching distance
is then defined as the supremum of all bottleneck distances
over all subspaces (see Section~\ref{sec:matching_distance} for the precise definition). 
The matching distance has been used 
in shape analysis ~\cite{cerri2013betti,italian} (where it is known as the matching distance
between size functions),
for virtual screening in 
computational chemistry~\cite{keller2018persistent},
and a recent algorithm~\cite{klo-exact} computes the distance exactly 
in polynomial time (with a large exponent).

Our contribution is an improvement of an approximation algorithm 
by Biasotti et al.~\cite{italian}
for the $2$-parameter case which we summarize next.
We parameterize the space of all lines of interest 
as a bounded rectangle $R\subset\R^2$. 
To each point $p$ in the rectangle,
we assign $f(p)$ as the bottleneck distance between the two persistence diagrams
when restricting the data sets to the line parameterized by $p$.
The matching distance is then equal to $\sup_{p\in R} f(p)$.
The major ingredient of the algorithm is a \emph{variation bound}
which tells how much $f(p)$ varies when $p$ is perturbed
by a fixed amount. For any subrectangle $S\subseteq R$ with center $c$,
$f(c)$ and the variation bound yield an upper bound of $f$ within $S$.
We then obtain an $\eps$-approximation with a simple branch-and-bound
scheme, subdividing $R$ with a quad-tree in BFS order and stopping the
subdivision of a rectangle when its upper bound is sufficiently small.

\subparagraph{Our contributions.}
We aim for a fast implementation
useful for practical applications of multi-parameter persistent homology. Towards this goal, we make the following contributions:
\begin{enumerate}
\item We rephrase the approximation algorithm by Biasotti et al. entirely
in elementary geometric terms.
We think that the geometric point of view complements
their formulation and makes the structure of the algorithm
more accessible. Indeed, we are able to simplify several arguments from~\cite{italian}~-- we defer a summary of these simplifications to 
\cref{app:comparison}.

\item We provide a simple yet crucial algorithmic
improvement: instead of using the global variation bound for all rectangles
of the subdivision, we derive adaptive local variation bounds for each
rectangle individually. This results in much smaller upper bounds
and avoids many subdivisions in the approximation algorithm.
\item We experimentally compare our version of the global
bound with the usage of the adaptive bounds. We show that the 
speed-up factor of the sharpest adaptive bound is typically between $3$ and $8$,
depending on the input bi-filtrations (for some inputs the speed-up is $15$).
\item Our code is available as part of \hera library\footnote{\url{https://bitbucket.org/grey_narn/hera/src/master/matching/}}
    and provides an efficient implementation for computing
    the matching distance for bi-filtrations.
\end{enumerate}

\subparagraph{Outline.}
We start with a short introduction to filtrations and persistent homology
in Section~\ref{sec:background}. We define the matching distance
in Section~\ref{sec:matching_distance}.
The algorithm to approximate it is described in Section~\ref{sec:approx_alg},
and our local variation bounds are derived in Section~\ref{sec:bound_primitive}.
We do experiments in Section~\ref{sec:experiments} and conclude
in Section~\ref{sec:conclusion}.

\section{Background}
\label{sec:background}
\subparagraph{Mono-Filtrations.}
Fixing a base set $V$, a \emph{$k$-simplex} $\sigma$ is a non-empty subset
of $V$ of cardinality $k+1$. A \emph{face} $\tau$ of $\sigma$
is a non-empty subset of $\sigma$. A \emph{simplicial complex} $K$
is a collection of simplices such that whenever $\sigma\in K$, every
face of $\sigma$ is in $K$ as well. The dimension of a simplicial complex $K$
is the maximal $k$ such that $K$ contains a $k$-simplex.
As an example, a graph is merely a simplicial complex of dimension $1$.
Following graph-theoretic notations, we call $0$-simplices of $K$ 
\emph{vertices} and $1$-simplices of $K$ \emph{edges}.
A \emph{subcomplex} of $K$ is a subset $L\subseteq K$ such that $L$ is
again a simplicial complex.
We will henceforth assume that the base set $V$ is finite, which implies
also that the simplicial complex $K$ is finite.

A \emph{mono-filtration} is a simplicial complex $K$ equipped with
a function $\ffunc:K\mapsto\R$ such that for any simplex $\sigma$
and any face $\tau$ of $\sigma$, it holds that $\ffunc(\tau)\leq\ffunc(\sigma)$.
We call $\ffunc(\sigma)$ the \emph{critical value of $\sigma$} 
and define for any $v\in\R$
\[
K_v:=\{\sigma\in K\mid\ffunc(\sigma)\leq v\}.
\]
By the condition on $\ffunc$ from above, $K_v$ is a subcomplex of $K$
for each $v$. Moreover, whenever $v\leq w$, we have that $K_v\subseteq K_w$.
Hence, the collection $(K_v)_{v\in\R}$ yields a nested sequence
of simplicial complexes, which is entirely determined by the critical values
of each simplex in $K$. See Figure~\ref{fig:mono-filtration} 
for an illustration.

\begin{figure}
\centering
\includegraphics[width=\textwidth]{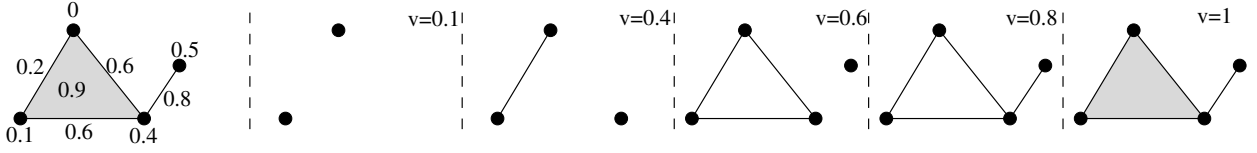}
\caption{Left: Mono-filtration of a simplicial complex $K$ of dimension $2$.
The critical value of each simplex is displayed. Right: Examples
of the complexes $K_v$ for various values of $v$.}
\label{fig:mono-filtration}
\end{figure}

\subparagraph{Persistence diagrams.}
We are interested in the topological changes of $(K_v)_{v\in\R}$
when $v$ increases continuously. A \emph{persistence diagram}
is a multi-set of points in $\R\times (\R\cup\{\infty\})$ 
with all points strictly  above the diagonal $x=y$. The general definition
requires a digression into representation theory and homological algebra
(e.g., see~\cite{oudot-book}).
Instead, we explain the idea on the problem of tracking connected
components of $K_v$ within the filtration, which is a special case of the
general theory. Assume for simplicity that no two simplices have the same
critical value. Whenever we reach the critical value $b$ of a vertex, a new
connected component comes into existence.
We call this a \emph{birth}. We say that the component born at $b$ \emph{dies}
at value $d$ if there is an edge with critical value $d$ that merges the
connected component with another connected component which was born before $b$.
In that case, $(b,d)$ is a point in the persistence diagram, denoting that
the corresponding connected component persisted from scale $b$ to scale $d$.
Assuming that $K$ is connected, each component gets assigned a unique
death value except the component born at the minimal critical value.
We assign the death value $\infty$ to this component, adding an infinite
point to the diagram.
The resulting diagram is called the 
\emph{persistence diagram in (homological) dimension $0$}.
See Figure~\ref{fig:pers_diag} (left).
Similar diagrams can be defined for detecting tunnels, voids,
and higher-dimensional holes in the simplicial complex.

\begin{figure}
\centering
\includegraphics[width=0.3\textwidth]{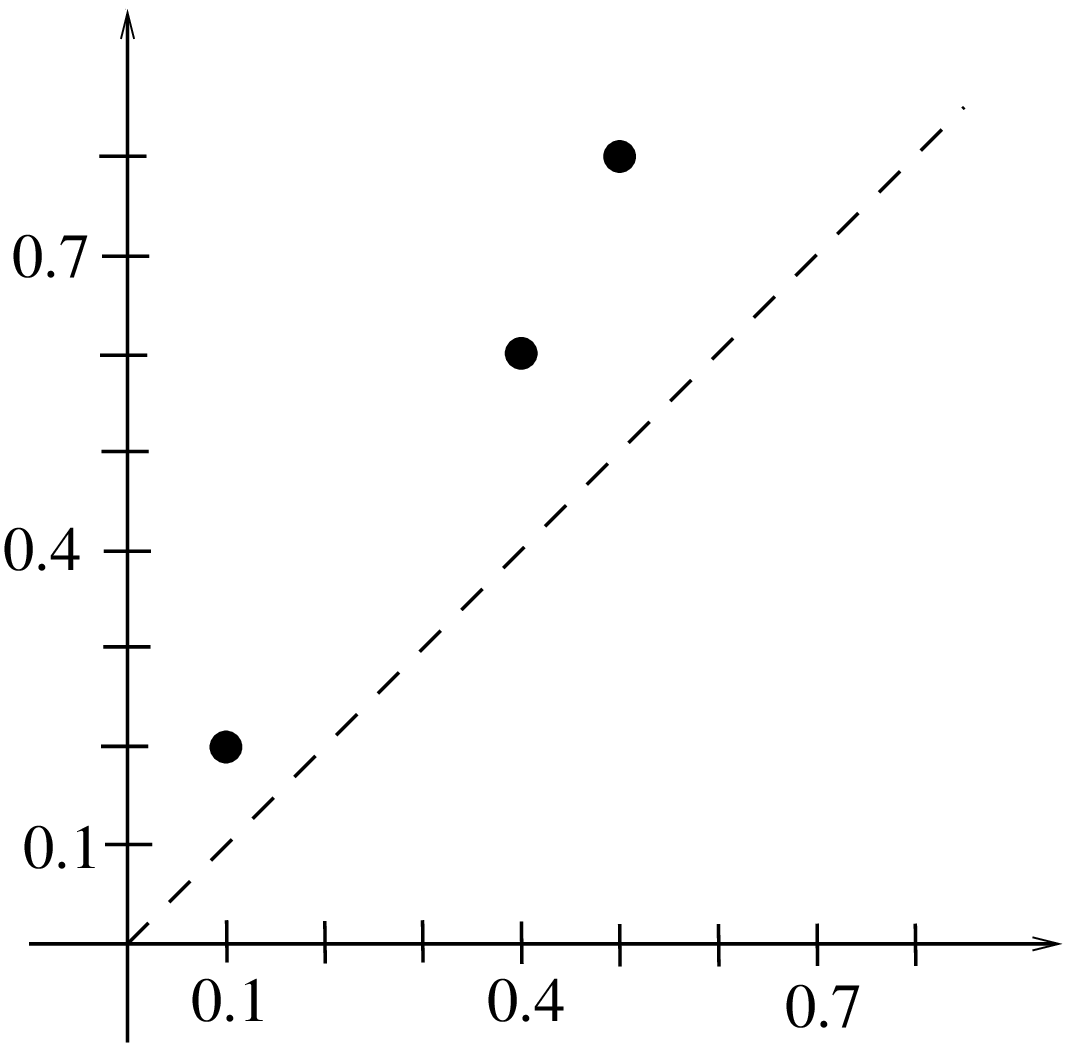}
\hspace{1cm}
\includegraphics[width=0.3\textwidth]{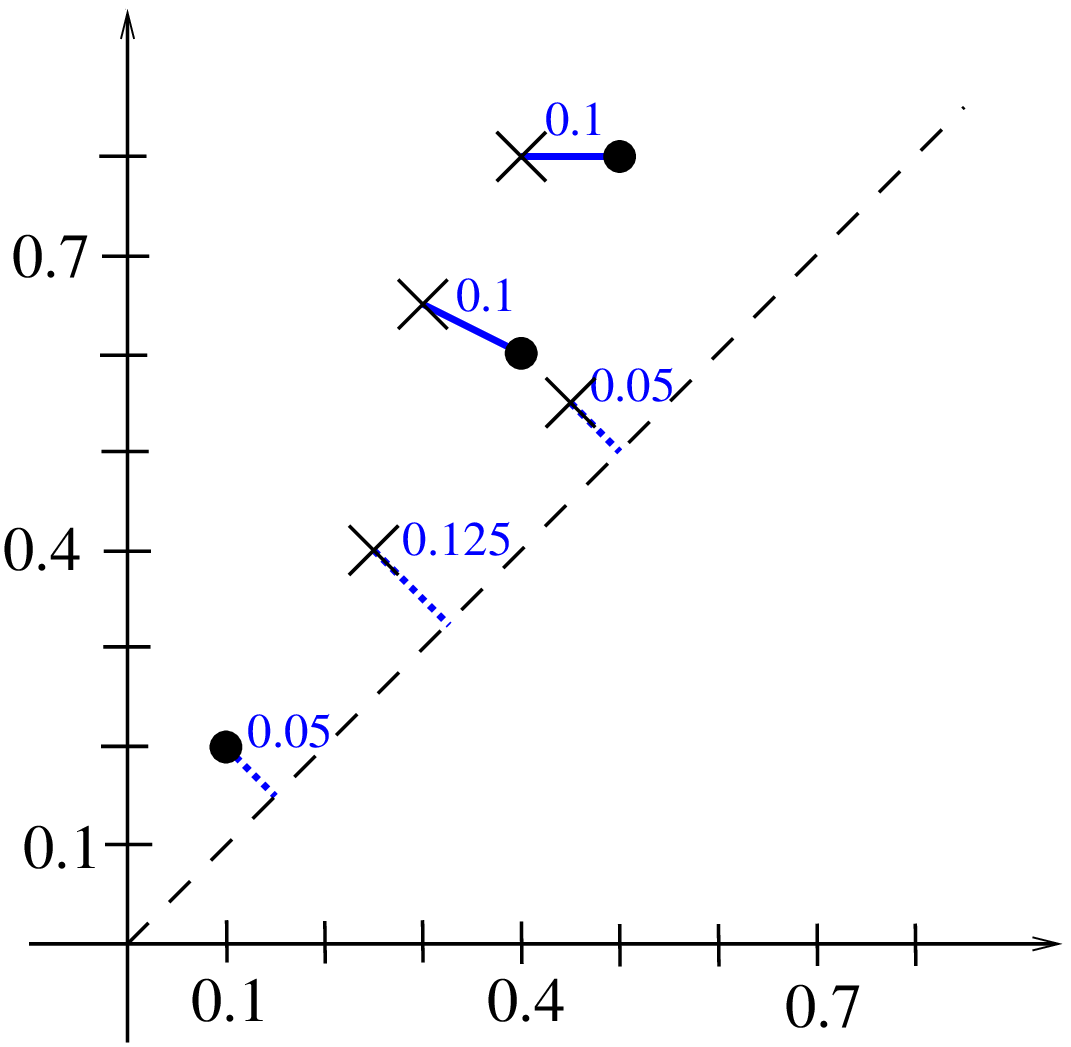}
\caption{Left: The persistence diagram in dimension $0$ of the example from Figure~\ref{fig:mono-filtration}
(plus the point $(0,\infty)$ that is not drawn).
Note that indeed, connected components are born at $0$, $0.1$, $0.4$ and $0.5$,
and the latter three components die at $0.2$, $0.6$ and $0.8$, respectively.
Right: A partial matching of two diagrams (depicted by circles and x-shapes).
The cost of each match and of each unmatched vertex is displayed. The cost
of this matching is $0.125$ which is in fact the optimal cost in this example,
so the bottleneck distance between the diagrams is $0.125$.}
\label{fig:pers_diag}
\end{figure}

We define a distance on two persistence diagrams $D_1$ and $D_2$ next. 
Fixing a partial matching between $D_1$ and $D_2$, we assign
to each match of $p\in D_1$ and $q\in D_2$ the cost 
$\|p-q\|_\infty=\max\{|p_x-q_x|,|p_y-q_y|\}$, with the understanding
that $\infty-\infty=0$. In particular, the cost of
matching a finite to an infinite point is $\infty$. Every unmatched point $p$
gets assigned the cost $\frac{p_y-p_x}{2}$ which corresponds to the
$L_\infty$-distance from $p$ to the diagonal. Taking the maximum over all matched
and unmatched points in $D_1$ and $D_2$ results in the cost of the 
chosen partial matching. The \emph{bottleneck distance} between $D_1$ and $D_2$
is then the minimum cost over all possible partial matchings between $D_1$
and $D_2$. 
See Figure~\ref{fig:pers_diag} (right) for an example.
Since filtrations give rise to persistence diagrams, 
we also talk about the bottleneck distance between two filtrations
and denote it by $d_B(\cdot,\cdot)$ from now on.

We will need the following properties of the bottleneck distance.
The proofs of the first three of them follow directly from the definition.
\begin{itemize}
\item $d_B$ satisfies the triangle inequality: $d_B(F,H)\leq d_B(F,G)+d_B(G,h)$
for three filtrations $F$, $G$, $H$.
\item $d_B$ is shift-invariant: let $F=(K,\ffunc)$ be a filtration, 
define $F_r$
be the filtration $(K,\ffunc+r)$, where the critical value of each simplex
is shifted by $r$. Then $d_B(F,G)=d_B(F_r,G_r)$.
\item $d_B$ is homogeneous: let $F=(K,\ffunc)$ be a filtration, $\lambda$ be a positive number,
define $\lambda F$
be the filtration $(K, \lambda \ffunc)$.
Then $d_B(\lambda F,\lambda G)= \lambda d_B(F,G)$.
\item $d_B$ is stable~\cite{ceh-stability}: let $F_1=(K,\ffunc_1)$ and $F_2=(K,\ffunc_2)$
be two filtrations of the same complex such that for each $\sigma\in K$,
$|\ffunc_1(\sigma)-\ffunc_2(\sigma)|\leq\eps$. Then, $d_B(F_1,F_2)\leq\eps$.
\end{itemize}

\subparagraph{Bi-filtrations.}
\label{par:bifil}
Define the partial order $\po$ on $\R^2$ as
$p\po q$ if and only if $p_x\leq q_x$ and $p_y\leq q_y$.
Geometrically, $p\leq q$ if and only if $q$ lies in the upper-right quadrant
with corner $p$. A \emph{($1$-critical) bi-filtration} is a simplicial
complex $K$ together with a function $\ffunc:K\to\R^2$ such that 
for every simplex $\sigma$ and every face $\tau$ of $\sigma$, we have
that $\ffunc(\tau)\po\ffunc(\sigma)$. As before, $\ffunc(\sigma)$ is called the
\emph{critical value} of $\sigma$. Our assumption
that every simplex has a unique critical value is 
just for the sake of simpler exposition; our ideas extend to the 
\emph{$k$-critical} case
where each $\sigma$ has up to $k$ incomparable critical values
(Appendix~\ref{app:k_critical}). Fixing $p\in\R^2$, we define
\[
K_p:=\{\sigma\in K\mid\ffunc(\sigma)\po p\}.
\]
Similar to the mono-filtration case, $K_p$ is
a subcomplex and $K_p\subseteq K_q$, whenever $p\po q$.

It is worth visualizing the construction of $K_p$ geometrically. 
We can represent the bi-filtration as a multi-set of points in $\R^2$,
where each point corresponds to a simplex and is placed at the critical
value of the simplex. The complex $K_p$ then consists of all simplices
that are placed in the lower-left quadrant with $p$ at its corner.
See Figure~\ref{fig:bi-filtration} for an example.

\begin{figure}
\centering
    \includegraphics[width=12cm]{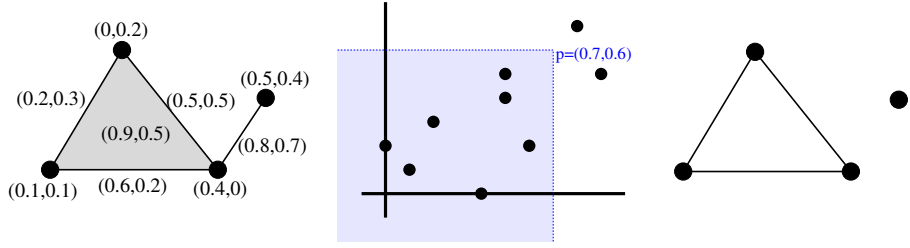}
\caption{Left: Bi-filtration of a simplicial complex $K$ of dimension $2$.
Middle: Every point in the plane denotes the critical value of a simplex.
The shaded rectangle yields the simplices that belong to $K_p$.
Right: Illustration of $K_p$ as a subcomplex of $K$.}
\label{fig:bi-filtration}
\end{figure}

\section{The matching distance}
\label{sec:matching_distance}

\subparagraph{Slices.}
Bi-filtrations are too wild 
to admit a simple combinatorial description such as a persistence diagram.
But we can obtain a persistence diagram when restricting to a 
one-dimensional affine subspace. 
For all concepts in this subparagraph, see Figure~\ref{fig:slice}
for an illustration.
We consider a 
non-vertical line $\slice$ with positive slope, which we call a \emph{slice}.
For every slice, we distinguish a point $\origin$, called the \textit{origin}
of the slice. We let $\slices$ denote the set of all slices.
Since the slope is positive, for any two distinct points $p$, $q$ on $\slice$
either $p\po q$ or $q\po p$ holds. Hence, $\po$ becomes a total order
along~$\slice$.

\begin{figure}
\centering
\includegraphics[width=12cm]{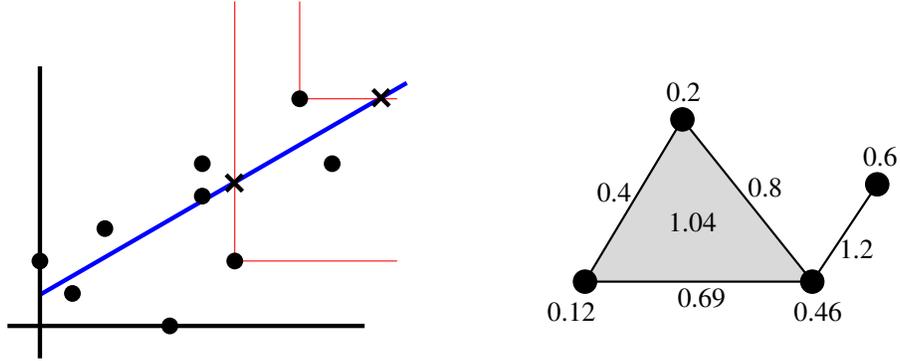}
\caption{Left: The slice parameterized by $(\frac{\pi}{6},0.1)$.
For two critical values of the bi-filtration from above,
we illustrate the construction of the point $q$ (displayed by a cross shape).
The push of the critical value is simply the Euclidean distance
to the point $(0,0.1)$, which is the origin of the slice.
Right: The non-weighted restriction on the slice. Each simplex
gets its push as critical value.}
\label{fig:slice}
\end{figure}

Given $p\in\R^2$, let $q$ be the minimal point on $\slice$ 
(with respect to $\po$) such that $p\leq q$.
Geometrically, $q$ is the intersection
of $\slice$ with the boundary of the upper-right quadrant of $p$,
or equivalently, the horizontally-rightwards projection of $p$
to $\slice$ if $p$ lies above $\slice$, or the vertically-upwards projection
of $p$ to $\slice$ if $p$ lies below $\slice$. Since $q$ lies on $\slice$,
$q$ can be written as 
\[\origin+\lambda_p\left(\begin{array}{c}\cos\sangle\\\sin\sangle\end{array}\right)\]
    where  $\sangle$ is the angle between $\slice$ and $x$-axis,
and $\lambda_p\in\R$.
We define $\lambda_p$ as the \emph{push} of
$p$ to $L$, which can be formally written as a function
$\push:\R^2\times\slices\to\R$. 
 Geometrically, the push is simply the (signed) distance
of the point $q$ to the origin of the slice.
Fixing a bi-filtration $F=(K,\ffunc)$,
the composition $\push(\cdot,\slice)\circ\ffunc$
yields a function $K\to\R$, and it can be readily checked that 
this function yields a mono-filtration, which we call the \emph{non-weighted restriction} of $F$ onto $\slice$.
See \cref{fig:slice} (right) for an example.

\subparagraph{Matching distance.}
Given two bi-filtrations $F^1$, $F^2$,
we could try to define a distance between
them by taking the supremum of bottleneck distances
between their non-weighted restrictions on all slices.
However, this
does not yield a meaningful result. The reason is that for almost
horizontal and almost vertical slices, the pushes of two close-by points
can move very far away from each other~-- see Figure~\ref{fig:unweighted_push_problem} for an example. 
As a result, the bottleneck distance
along such slices becomes arbitrarily large.

\begin{figure}
\centering
\includegraphics[width=5cm]{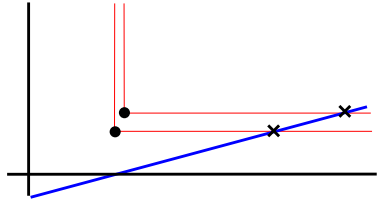}
\caption{Two points that are close to each other might have pushes
far from each other. Note that by making the
slice more flat, the distance between the pushes can be made arbitrarily large.}
\label{fig:unweighted_push_problem}
\end{figure}

Instead, we introduce a weight for each slice.
Let $\sangle$ denote the angle between the slice $\slice$
and $x$-axis.
We call $\slice$
\emph{flat} if $\sangle\leq\frac{\pi}{4}$
(i.e., if its slope is $\leq 1$)
and \emph{steep} if $\sangle\geq\frac{\pi}{4}$. Then we set
\[
\weight(\slice):=\begin{cases}
\sin\sangle& \text{if $\slice$ is flat}\\
\cos\sangle& \text{if $\slice$ is steep}.
\end{cases}
\]

We define the matching distance between the bi-filtrations
$F^1=(K^1,\ffunc^1)$ and $F^2=(K^2,\ffunc^2)$ as
\[
    \dmatch(F^1,F^2):=\sup_{\slice\in\slices} \weight(\slice) \cdot d_B(restr(F^1, \slice), restr(F^2, \slice)),
\]
where $restr(F^i, \slice)$
denotes the non-weighted restriction of $F^i$ onto $\slice$.
Note that while the non-weighted restrictions depend on the choice of the
origin, a different choice of origin for a slice only results in a 
uniform translation of the critical values of both mono-filtrations.
Hence, the bottleneck distance does not change because of shift-invariance.
This means that the matching distance is independent 
of the choice of the origins.

Moreover, the shift-invariance of $d_B$ implies that if we alter $\ffunc_1$
and $\ffunc_2$ such that every value is translated 
by the same vector $v\in\R^2$,
the matching distance does not change. Recalling that we can visualize
bi-filtrations as finite multi-sets of points in $\R^2$, we can hence assume
without loss of generality that all these points are in the upper-right
quadrant of the plane, that is, $\ffunc^i(\sigma)\in[0,\infty)\times [0,\infty)$.

Let us now define the \textit{weighted push} of a point $p$
to a slice $\slice$
as $\wpush(p, {\slice}) = \weight(\slice) \push(p, \slice)$,
and let $F_\slice$ denote the mono-filtration
induced by $\sigma \mapsto \wpush(\ffunc(\sigma), \slice)$. 
We call $F_{\slice}$ a \emph{weighted restriction} of $F$ onto $\slice$.
Note that $F_\slice$ equals $restr(F, \slice)$ except that
all critical values are scaled by the factor $\weight(\slice)$.
Using homogeneity of $d_B$, we see that
\begin{eqnarray}
    \dmatch(F^1, F^2) = \sup_{\slice\in\slices} d_B(F^1_{\slice}, F^2_\slice).
\label{eqn:matching}
\end{eqnarray}
We will use this equivalent definition of the matching distance in the remaining
part of the paper, and 'restriction' will always mean 'weighted restriction'.

\section{The approximation algorithm}
\label{sec:approx_alg}
The idea of the approximation algorithm for $\dmatch$ is to sample
the set of slices through a finite sample, and chose the maximal
bottleneck distance between the (weighted) restriction encountered
as the approximation value. In order to execute this plan, 
we need to parameterize the space of slices and need to compute 
the restriction of a parameterized slice efficiently.

\subparagraph{Slice Parameterization.}
Every slice has a unique point where the line enters the
positive quadrant of $\R^2$, which is either its intersection with the
positive $x$-axis, the positive $y$-axis, or the point $(0,0)$.
From now on, we always use this point as the origin of the slice.

\begin{figure}
\centering
\includegraphics[width=5cm]{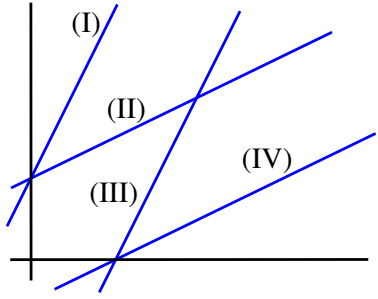}
\caption{A steep $y$-slice (I), a flat $y$-slice (II), a steep $x$-slice(III) and a flat $x$-slice. The slopes are $2$ for the steep and $\frac{1}{2}$ for
the flat slices, and the origin is at $(0,2)$ for the $y$-slices and at $(2,0)$ for the $x$-slices. Consequently, all four slices are parameterized by $(\frac{1}{2},2)$.}
\label{fig:slice_types}
\end{figure}

We call a slice an \emph{$x$-slice} if its origin lies on the positive $x$-axis, and call it a \emph{$y$-slice} if its origin lies on the positive
$y$-axis (slices through the origin are both $x$- and $y$-slices).
Recall also that a slice is flat if its slope is less than $1$,
and steep if it is larger than $1$.
Thus, a slice belongs to one of the four types:
flat $x$-slices, flat $y$-slices, steep $x$-slices
and steep $y$-slices. 
Every slice is represented as a point $(\sslope, \smu)\in (0,1]\times [0,\infty)$ where the interpretation of the parameters depends on the type of the slice
as follows: Let $\origin=(\origin_x,\origin_y)$ be the origin of $\slice$,
and recall that $\sangle$ is the angle of the slice with the
$x$-axis. Then
\[
    \sslope = \begin{cases}
        \tan(\sangle),\mbox{ if $\slice$ is flat}\\
        \cot(\sangle),\mbox{ if $\slice$ is steep}
    \end{cases},\quad\quad
\smu = \begin{cases}
        \origin_x &\text{if $\slice$ is $x$-slice}\\
        \origin_y &\text{if $\slice$ is $y$-slice}
    \end{cases}.
\]
In other words, $\sslope$ is the slope of the line in the flat case, and
the inverse of the slope in the steep case, and $\smu$ contains the non-trivial
coordinate of the origin.
Note that the same pair of parameters can parameterize different slices
depending on the type. Figure~\ref{fig:slice_types} illustrates this.

\subparagraph{Weighted pushes.}
We next show a simple formula for the value of $\wpush(p,\slice)$
depending on the type of the slice.
\begin{mylem}
With the chosen parameterization and choice of origin on $\slice$, 
    $\wpush(p,\slice)$
is computed according to the formulas given 
in \cref{tbl:wpush-formulas}.
\label{lem:wpush}
\end{mylem}
\begin{proof}
The proof of the lemma is a series of elementary calculations.
Let us consider, for example, the case of a flat $y$-slice.
In this case the slice $\slice = (\sslope, \smu)$ is given by
\[
\left\{\colpt{0}{\mu}+\rho\colpt{\cos\sangle}{\sin\sangle}\mid\rho\in\R\right\},
\]
If $p=(p_x,p_y)$ is above $\slice$,
we consider the point $q = (q_x, q_y)$ which is the intersection of $\slice$
and the line $y = p_y$.
Obviously, $q_y = p_y$, and, since $q$ lies
on $\slice$, the second coordinate yields
\[\rho=\frac{q_y-\mu}{\sin\sangle}=\frac{p_y-\mu}{\sin\sangle}.\]
By definition, $\rho$ is the push of $p$ to $\slice$
and since $\slice$ is flat, we have that $\weight(\slice)=\sin\sangle$
which cancels with the denominator. The other $7$ cases are proved
analogously.
%
\end{proof}

\begin{table}
    \ra{1.2}
\begin{tabular}{@{}ccccccc@{}}\toprule
  \phantom{p above L} &   & \multicolumn{2}{c}{$y$-slices} & \phantom{a} & \multicolumn{2}{c}{$x$-slices}\\
    \cmidrule{3-4}  \cmidrule{6-7}
   & & flat & steep                  &  & flat & steep\\
    $p$ above $\slice$   & &  $ p_y - \smu$  & $\sslope(p_y - \smu)$&   & $ p_y $           & $\sslope p_y$  \\ 
    $p$ below $\slice$   & &  $\sslope p_x $  & $p_x$    &     &  $\sslope (p_x - \smu) $  & $p_x - \smu$ \\
\bottomrule
\end{tabular}
    \caption{Formulas for weighted push of $(p_x, p_y)$ onto a slice $\slice = (\sslope, \smu)$.}
\label{tbl:wpush-formulas}
\end{table}

%
All 8 expressions in \cref{tbl:wpush-formulas}
involve only addition and multiplication without trigonometric functions.
Hence we can extend them continuously to $\sslope=0$, which
corresponds to horizontal lines (in the flat case)
or vertical lines (in the steep case). 
With this interpretation, 
we can extend $\slices$ to a set $\bar\slices$ in (\ref{eqn:matching}),
containing these limit cases, without changing the supremum.

Next, we observe that we can restrict our attention to a bounded range
of $\smu$-parameters. For that, let $X$ denote the maximal $x$-coordinate and
$Y$ be the maximal $y$-coordinate among all critical values
of $F^1$ and of $F^2$. 
For a $y$-slice (steep or flat) $\slice = (\sslope, \smu)$ with $\smu > Y$,
let $\slice' = (\sslope, Y)$ be the parallel slice with origin at $(0,Y)$.
All critical values of $F^{1,2}$ are below $\slice$ and $\slice'$
by construction (recall that all critical points are assumed 
in the upper-right quadrant),
hence we obtain the push by projecting vertically upwards. 
Looking at the second row in \cref{tbl:wpush-formulas},
we see that the weighted pushes are independent of $\smu$,
and therefore equal for $\slice$ and $\slice'$. 
Hence, the weighted bottleneck
distances along $\slice$ and $\slice'$ are equal: $d_B(F^1_{\slice}, F^2_{\slice}) = d_B(F^1_{\slice'}, F^2_{\slice'})$.
See Figure~\ref{fig:slice_space_proof_illu} for an illustration.
We conclude that for $y$-slices it suffices to consider $0 \leq \smu \leq Y$
in (\ref{eqn:matching}) without changing the matching distance.
An analogous argument shows that for $x$-slices, it is only necessary
to consider $0 \leq \smu \leq X$.

\begin{figure}
\centering
\includegraphics[width=4cm]{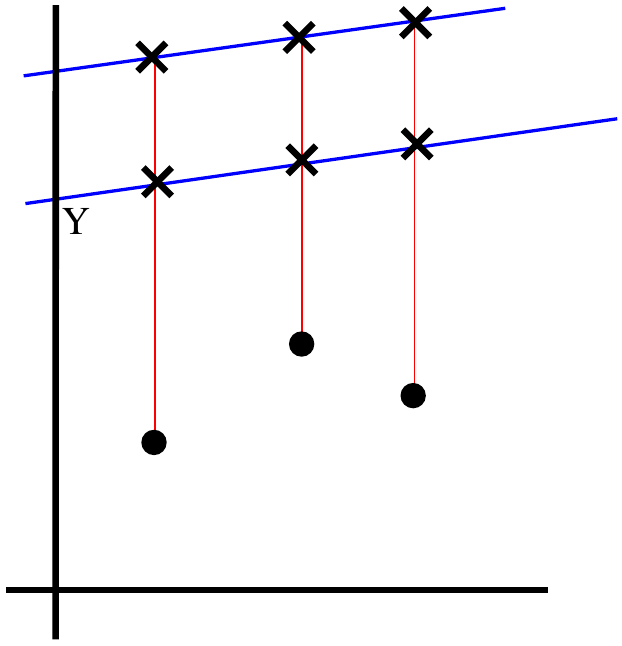}
\caption{Illustration for the fact that slices with larger value of $\smu$
can be ignored.}
\label{fig:slice_space_proof_illu}
\end{figure}

To summarize the last two observations, we arrive at the following statement.
There are sets $\slices_1$ of flat $x$-slices, $\slices_2$ of steep $x$-slices,
$\slices_3$ of flat $y$-slices, and $\slices_4$ of steep $y$-slices
(with each set containing some vertical/horizontal lines as limit case)
such that
\begin{eqnarray}
\dmatch(F^1, F^2) = \sup_{\slice\in\slices_1\cup\ldots\cup\slices_4} d_B(F^1_{\slice}, F^2_\slice)
\label{eqn:matching_compact}
\end{eqnarray}
and such that $\slices_1$ and $\slices_2$ are parameterized by 
$[0,1]\times[0,X]$ and $\slices_3$ and $\slices_4$ by
$[0,1]\times[0,Y]$. 

\subparagraph{Approximation.}
We present an approximation algorithm that, given two bi-filtrations
$F^1$ and $F^2$ and some $\eps>0$ returns a number $\delta$ such that
\[
\dmatch(F^1,F^2)-\eps \leq \delta\leq \dmatch(F^1,F^2).
\]
We assume that the two bi-filtrations are given as simplicial complexes,
i.e., a list of simplices, where each simplex is annotated with
two real values denoting the critical value of the simplex.
In the description, we set $T:=\{\text{$x$-flat}, \text{$x$-steep}, \text{$y$-flat}, \text{$y$-steep}\}$ for the type of a slice.
The algorithm is based on the following two primitives:
\begin{description}
\item[\texttt{Eval($F^1, F^2, \slice$)}] Computes $d_B(F^1_\slice,F^2_\slice)$,
where $\slice$ is specified by the triple $(\sslope,\smu,t)$
where $(\sslope,\smu)$ are the parameterization of $\slice$ and $t\in T$
denotes its type.
\item[\texttt{Bound($F^1, F^2, \slicebox,t$)}] 
If $\slicebox$ is an axis-parallel rectangle and $t\in T$,
the pair $(B,t)$ specifies a set of slices $\slices_0$. 
The primitive computes a number $\mu\in\R$ such that, for every $\slice\in\slices_0$, 
\[d_B(F^1_\slice,F^2_\slice)\leq\mu.\]

\end{description}

With these two primitives, we can state our approximation algorithm:
from now on, we refer to axis-parallel rectangles 
as \emph{boxes} for brevity.
We start by computing maximal coordinates $X$ and $Y$
of critical values of $F^1$ and $F^2$
and enqueueing the four initial items $([0,1]\times [0,Y], \text{$y$-steep})$,
$([0,1]\times [0,Y], \text{$y$-flat})$,
$([0,1]\times [0,X], \text{$x$-steep})$, and
$([0,1]\times [0,X], \text{$x$-flat})$
into a FIFO-queue.
We also maintain a variable $\rho$ storing the largest bottleneck distance
encountered so far, initialized to $0$.

Now, we pop items from the queue and repeat the following steps:
for an item $(\slicebox,t)$, let $\slice$ denote the slice that corresponds to the center point of $\slicebox$. We call \texttt{Eval($F^1$,$F^2$,$\slice$)}
and update $\rho$ if the computed value is bigger than the current maximum.
Then, we compute $\mu\gets$\texttt{Bound($F^1$,$F^2$,$\slicebox$,$t$)}.
If $\mu>\rho+\eps$, we split $\slicebox$ into $4$ sub-boxes
$\slicebox_1,\ldots,\slicebox_4$
of equal dimensions (using the center as splitpoint)
and enqueue $(\slicebox_1,t),\ldots,(\slicebox_4,t)$. 
When the queue is empty,
we return $\delta\gets \rho$. This ends the description
of the algorithm.

Assuming that the above algorithm terminates (which is unclear at this point
because it depends on the implementation of the \texttt{Bound} primitive),
the output is indeed an $\eps$-approximation.
This can be derived directly from the termination condition
of the subdivision and the fact that $\rho$ is non-decreasing
during the algorithm. See Appendix~\ref{app:correctness} for details.

A variant of the above algorithm computes a relative approximation of
the matching distance, that is, a number $\delta$ such that
\[
\dmatch(F^1,F^2) \leq \delta\leq (1+\eps)\dmatch(F^1,F^2).
\]
The algorithm is analogous to the above, with the difference that
a box is subdivided if $\mu >(1+\eps)\rho$,
and at the end of the algorithm $(1+\eps)\rho$ is returned as $\delta$.
The correctness of this variant follows similarly. 
However, 
the algorithm terminates only if $\dmatch(F^1,F^2)>0$,
and its complexity depends on the value of the matching distance.

\medskip

What is needed to realize the \texttt{Eval} primitive? First, we compute the weighted pushes of each critical value 
of $F^1$ and of $F^2$ in time proportional to the number
of critical values using Lemma~\ref{lem:wpush}.
Then, we compute the persistence diagrams of $F^1$
and of $F^2$, and their bottleneck distance.
Both steps are well-studied standard tasks in persistent homology,
and several practically efficient algorithms have been studied. We use \textsc{Phat}~\cite{bkrw-phat}
for computing persistence diagrams and \textsc{Hera}~\cite{kmn-geometry}
for the bottleneck computation.

\section{The \texttt{Bound} primitive}
\label{sec:bound_primitive}
Recall that the input of \texttt{Bound} is $(F^1,F^2,\slicebox,t)$,
where $(\slicebox,t)$ specifies a collection of slices
of type $t$. In what follows, we will identify points in $\slicebox$
with the parameterized slice, writing $\slice\in\slicebox$
to denote that $\slice$ is obtained from a pair of parameters $(\sslope,\smu)\in\slicebox$ with respect to type $t$ (which we skip for notational convenience).

Let $\centerslice$ be the slice corresponding to the center of $\slicebox$.
The \emph{variation} of a point $p\in\R^2$ for $\slicebox$ denotes
how much the weighted push of $p$
changes when the slice is changed within $\slicebox$:
\[
    v(p,\slicebox):=\max_{\slice\in\slicebox} |\weightedpush_p(\slice)-\weightedpush_p(\centerslice)|.
\]
For a bi-filtration $F$, we define
\[
v(F,\slicebox):=\max_{\text{$p$ critical value of $F$}} v(p,\slicebox).
\]
The variation yields an upper bound
for the bottleneck distance within a box:
\begin{mylem}\label{lem:upper_bound_lemma}
With the notation as before, we have that for two filtrations $F^1$, $F^2$
that
\[
\sup_{\slice\in\slicebox}d_B(F^1_{\slice},F^2_{\slice})\leq v(F^1,\slicebox)+d_B(F^1_{\centerslice},F^2_{\centerslice})+v(F^2,\slicebox)
\]
\end{mylem}
\begin{proof}
By triangle inequality of the bottleneck distance,
\[
d_B(F^1_{\slice},F^2_{\slice})\leq d_B(F^1_{\slice},F^1_{\centerslice})+d_B(F^1_{\centerslice},F^2_{\centerslice})+d_B(F^2_{\centerslice},F^2_{\slice}).
\]
Looking at the first term on the right, we have two filtrations of the
same simplicial complex, and every critical values changes by at most $v(F^1,\slicebox)$ by definition of the variation. Hence, by stability of the bottleneck
distance, $d_B(F^1_{\slice},F^1_{\centerslice})\leq v(F^1,\slicebox)$.
The same argument applies to the third term which proves the theorem.
\end{proof}
Note that the second term in the bound of Lemma~\ref{lem:upper_bound_lemma}
is the value at the center slice, which is already computed in the algorithm.
It remains to compute the variation of a bi-filtration within $\slicebox$.
This, in turn, we do by analyzing the variation of a point $p$ within $\slicebox$. We show

\begin{thm}\label{thm:corner-point}
For a box $\slicebox$, let $L_1,\ldots,L_4$ be the four slices
on the corners of $\slicebox$. Then
\[
v(p,\slicebox)=\max_{i=1,\ldots,4}|\weightedpush_p(\slice_i)-\weightedpush_p(\centerslice)|
\]
\end{thm}

The theorem gives a direct algorithm to compute $v(p,\slicebox)$,
just by computing the weighted pushes at the four corners 
(in constant time) and return
the maximal difference to the weighted push in the center.
Doing so for every critical point of a bi-filtration $F$ yields
$v(F,\slicebox)$, and with Lemma~\ref{lem:upper_bound_lemma}
an algorithm for the \texttt{Bound} primitive that runs
in time proportional to the number of critical points of $F^1$ and $F^2$.
We refer to this bound as \emph{local linear bound} (where the term
``linear'' refers to the computational complexity), or as \emph{L-bound}.

The proof of the theorem is presented in Appendix~\ref{app:proof_hor_vert}.
The main idea is that the expression $|\wpush_p(\slice_{\sslope,\smu})-\wpush_p(\centerslice)|$ 
(with $\slice_{\sslope,\smu}$ the slice given by $(\sslope,\smu)$)
has no isolated local maxima, even for a fixed $\sslope$
or a fixed $\smu$. That implies that from any $(\sslope,\mu)$
in the box, there is rectilinear path to a corner on which the expression
above is non-decreasing.

\subparagraph{A coarser bound.}
We have derived a method to compute $v(p,\slicebox)$ exactly which takes
linear time. Alternatively, we can derive an upper bound as follows:

\begin{thm}\label{thm:local_constant_bound}
    Let $\slicebox$ be a box $[\sslope_{\min}, \sslope_{\max}] \times [\smu_{\min}, \smu_{\max}]$
with center $(\centersslope,\centersmu)$, width $\Delta \sslope = \sslope_{\max} - \sslope_{\min}$  and height $\Delta \smu = \smu_{\max} - \smu_{\min}$.
Then, for any point $p\in[0,X]\times[0,Y]$, $v(p,\slicebox)$ is at most
\[
    \begin{cases}
        \frac{1}{2}\left(\Delta \smu+ X \Delta \sslope \right) & \mbox{for flat $y$-slices}\\
        \frac{1}{2}\left(\centersslope \Delta \smu + (Y - \smu_{\min}) \Delta \sslope \right)\} & \mbox{for steep $y$-slices}\\
        \frac{1}{2}\left(\sslope_{c} \Delta \smu + (X - \smu_{\min}) \Delta \sslope \right) & \mbox{for flat $x$-slices}\\
        \frac{1}{2}\left(\Delta \smu+ Y \Delta \sslope \right) & \mbox{for steep $x$-slices}.
     \end{cases}
\]
\end{thm}
Importantly, the bound is independent of $p$, and hence
also an upper bound for $v(F,\slicebox)$ that can be computed in constant time;
we refer to it as \emph{local constant bound} or \emph{C-bound}.

The proof of Theorem~\ref{thm:local_constant_bound} is based on deriving
a bound of how much $\wpush_p(\slice)$ and $\wpush_p(\slice')$ can differ
for two slices $\slice=(\sslope,\smu)$ and $\slice' = (\sslope',\smu')$
in dependence of $|\sslope-\sslope'|$ and $|\smu-\smu'|$. This bound,
in turn, is derived separately for all four types of boxes
and involves an inner case distinction depending on whether $p$ lies
above both slices, below both slices, or in-between. In either case,
the claim of the statement follows from the bound by plugging in the
center slice of a box for either $\slice$ or $\slice'$. 
See Appendix~\ref{app:constant_bound} for the detailed proof.

\subparagraph{Termination and complexity.}
We show next that our absolute approximation algorithm terminates when realized
with either the local linear bound or the local constant bound.
In what follows, set $C:=\max\{X,Y\}$.
In the subdivision process, each box $\slicebox$ considered is assigned
a \emph{level}, where the level of the four initial boxes is $0$,
and the four sub-boxes obtained from a level-$k$-box have level $k+1$.
Since every box is subdivided by the center, we have
immediately that for a level-$k$-box, $\Delta\sslope=2^{-k}$,
and $\Delta\smu\leq C 2^{-k}$. Using these estimates in 
Theorem~\ref{thm:local_constant_bound},
we obtain
\begin{equation}
v(F^i,\slicebox)\leq \frac{1}{2}(C2^{-k}+C2^{-k})=C2^{-k}.
\label{eqn:glob_bound}
\end{equation}
for $i=1,2$ and every level-$k$-box $\slicebox$ considered by the algorithm.
Note that the local constant bound yields a bound on $v(F^i,\slicebox)$
that is not worse, and so does the local linear bound (which computes
$v(F^i,\slicebox)$ exactly). Hence we have

\begin{mylem}
Let $\slicebox$ be a level-$k$-box considered in the algorithm. 
Then, $\mu$, the result of the \texttt{Bound} primitive in the algorithm,
satisfies
\[\mu\leq d_B(F^1_{\centerslice},F^2_{\centerslice})+2C2^{-k}\]
both for the local linear and local constant bound.
\end{mylem}

It follows easily that if $2C2^{-k}\leq\eps$,
or equivalently $2^k\geq \frac{2C}{\eps}$, a box is not further subdivided
in the algorithm (see Lemma~\ref{lem:k_bound} in Appendix~\ref{app:complexity}). 
Moreover, if the maximal subdivision depth is $k$, the algorithm
visits $O(4^k)$ boxes, and requires $O(n^3)$ time per box because
of the computation of two persistence diagrams and their bottleneck distance.
Combining these results leads to the complexity bound.

\begin{thm}\label{thm:absolute_approx_complexity}
Our algorithm to compute an absolute $\eps$-approximation terminates in 
\[
O(n^3 \left(\frac{C}{\eps}\right)^2)
\]
steps in the worst case (both for the linear and constant bound).
\end{thm}

See again Appendix~\ref{app:complexity} for more details. In there,
we also derive a similar bound for the variant of computing a relative
approximation. 

\begin{thm}\label{thm:relative_approx_complexity}
Our algorithm to compute a relative $(1+\eps)$-approximation terminates in
\[
O(n^3\left(\frac{C(1+\eps)}{\eps\dmatch(F^1,F^2)}\right)^2)
\]
steps in the worst case if $\dmatch(F^1,F^2)>0$. 
\end{thm}

\section{Experiments}
\label{sec:experiments}

\subparagraph{Experimental setup.}
Our experiments were performed on a workstation with an Intel(R) Xeon(R) CPU E5-1650 v3 CPU (6 cores,
3.5GHz) and 64 GB RAM, running GNU/Linux (Ubuntu 16.04.5). 
The code was written in C++ and compiled with gcc-8.1.0.

We generated two datasets, which we call \dsGH and \dsED,
following \cite{italian} (unfortunately,
we were unable to get either the code or the data used by the authors).
Each of the 70 files in the datasets is a lower-star bi-filtration of a triangular mesh (2-dimensional complex),
representing a 3D shape.
We also generated dataset \dsR of larger random bi-filtrations with up to {2,000} vertices.
A more detailed description of the datasets can be found in \cref{sec:datasets}.
In all our experiments we used persistence diagrams in dimension 0.
In the experiments with the datasets \dsGH and \dsED, we computed all pairwise distances;
in the experiments with \dsR we computed distances only between bi-filtration with the same number of vertices.
We used relative error threshold, which we call $\eps$ in this section
throughout (i.e., we always compute $(1+\eps)$-approximation).

\subparagraph{Comparison of different bounds.}
First, we experimentally compare the performance
of our algorithm with the L-bound from Theorem~\ref{thm:corner-point} and with
the C-bound from Theorem~\ref{thm:local_constant_bound}.
Obviously, the L-bound is sharper, so it allows us to subdivide
fewer boxes and in this sense is more efficient.
However, it is not a priori clear that the L-bound
is preferable, since its computation takes $O(n)$ time per box, in contrast
to the constant bound.

\ignore{
There are two natural optimizations of the algorithm with the L-bound.
Note that the value returned by \texttt{Bound} is used only to decide
whether we have to subdivide the current box: we subdivide, if
the bound is less than $(1+ \eps)\rho$, where $\rho$ is the lower bound,
the maximum of all $d_B(F^1_L, F^2_L)$ we have computed so far.
The first optimization is to compute the C-bound, and, if it is less than $(1+\eps)\rho$,
stop subdividing (the L-bound cannot be greater than the C-bound).
Otherwise we range over all critical values $p$, computing their variation.
As second optimization, we stop once we encounter a point $p$ whose variation
is large enough to ensure that the L-bound
will be greater than $(1+\eps)\rho$.
}

Secondly, we compare the local bounds L and C with the bound
provided by \cref{eqn:glob_bound}, which we call the \emph{global} bound or
G-bound because it only depends on the size of the box.
A bound of that sort (worse than \cref{eqn:glob_bound})
is used in~\cite{italian}.

Recall that the dominating step in the complexity analysis (and in practice) is the computation of persistence,
and we perform two such computations when we call the \texttt{Eval} primitive.
Therefore we are interested in the number of calls of \texttt{Eval};
for brevity, we refer to this number simply as the number of \textit{calls}.

In \cref{tbl:avg_calls_time}, we give the average
number of calls and timings for different datasets and values of $\eps$.
Actually, the variance behind the average in these tables is large,
so we additionally provide \cref{tbl:call_ratio}
and  \cref{tbl:time_ratio},
where we report the average, maximal, and minimal ratios
of the number of calls and time that the algorithm needs with different bounds.
For instance, the third line of \cref{tbl:call_ratio} shows that for all pairs 
from \dsED that we tested with relative error $\eps = 0.5$, switching from the local constant to the local linear
bound reduces the number of calls by a factor between 1.78 and 4.92.

\cref{tbl:time_ratio} shows that, as expected, the C-bound always performs better
than the G-bound, with the average speed-up around $2$.
The L-bound brings an additional speed-up by a factor of $1.5$-$2$ in terms of the running time;
the number of calls is reduced more significantly, by a factor of $3$.
However, the second from the right column of \cref{tbl:time_ratio} shows that the running
time can sometimes moderately increase, if we switch to the L-bound from the C-bound.
If we compare the G-bound with the L-bound directly (these numbers
are not present in \cref{tbl:time_ratio}), the best speed-up factor is $15.6$,
the worst one is $1.14$, and the average is between $3$ and $8$, depending on the dataset.

\begin{table}\centering
    \ra{1.2}
\begin{tabular}{@{}rrrrrrr@{}}\toprule
    & \multicolumn{3}{c}{\#Calls} & \multicolumn{3}{c}{Time (min)} \\
                          \cmidrule{2-4}  \cmidrule{5-7} 
                          & L & C & G & L & C & G  \\ \midrule
\dsGH, $\eps = 0.5$  & 938 & 2502 & 11082   & 2.08 & 3.67 & 18.78 \\
\dsED, $\eps = 0.1$  & 1455 & 3920 & 27529  & 2.58 & 3.26 & 25.96 \\
\dsED, $\eps = 0.5$  & 169 & 531 & 2112     & 0.28 & 0.42 & 1.67 \\
 \bottomrule
\end{tabular}
    \caption{Average number of calls and average running time with the L-, C- and G-bounds for different datasets and relative error $\eps$.}
\label{tbl:avg_calls_time}
\end{table}

\begin{table}\centering
    \ra{1.2}
\begin{tabular}{@{}rrrrrrr@{}}\toprule
    & \multicolumn{3}{c}{ \#Calls: G / C} & \multicolumn{3}{c}{ \#Calls: C / L} \\
                          \cmidrule{2-4} \cmidrule{5-7} 
 Dataset, $\eps$     & Avg & Min & Max     & Avg & Min & Max \\ \midrule
\dsGH, $\eps = 0.5$  &    1.80 &   1.21 &   3.28 &   3.10 & 1.51 & 7.02 \\
\dsED, $\eps = 0.1$  &    2.93 &   1.43 &   5.07 &   3.00     & 1.88 & 6.82 \\
\dsED, $\eps = 0.5$  &    1.94 &   1.17 &   2.81 &   3.29 &  1.78 & 4.92\\
\dsR, $\eps = 0.1$  &    6.06 &   2.76 &  10.58 &   2.08 &   1.91 &   2.47 \\ 
\bottomrule
\end{tabular}
    \caption{Comparison of number of calls between the global, local constant, and local linear bounds.
    G / C denotes the ratio of the G-bound compared with the C-bound;
    C / L denotes the ratio of the C-bound compared with the L-bound.}
\label{tbl:call_ratio}
\end{table}

\begin{table}\centering
    \ra{1.2}
\begin{tabular}{@{}rrrrrrr@{}}\toprule
                          & \multicolumn{3}{c}{Time: G / C} & \multicolumn{3}{c}{Time: C / L}\\
                          \cmidrule{2-7}
Dataset, $\eps$      & Avg & Min & Max        & Avg & Min & Max  \\ \midrule
\dsGH, $\eps = 0.5$  &    1.66 &   1.00 &   3.18 &    2.03 & 0.75 & 6.32 \\ 
\dsED, $\eps = 0.1$  &    3.12 &   1.44 &   5.21 &    1.64 &  0.81 & 3.40 \\
\dsED, $\eps = 0.5$  &    2.08 &   1.07 &   3.38 &    1.59 & 0.92 & 3.89 \\
\dsR, $\eps = 0.1$  &    5.73 &   2.83 &  10.67 &   1.93 &   1.66 &   2.20 \\
\bottomrule
\end{tabular}
    \caption{Comparison of running time between the global, local constant, and local linear bounds.}
\label{tbl:time_ratio}
\end{table}

\subparagraph{Breadth-first search, depth-first search, and error decay.}
\label{para:comp_error_decay}
In the formulation of our algorithm we used a FIFO-queue. This means
that we traverse the quad-tree in breadth-first order
(i.e., level by level).
Other traversal strategies are possible, for instance
a depth-first order, or a greedy algorithm
where boxes with large bottleneck distance at the center are picked first.
We experimented with these variants
and found no significant difference.  The explanation is that a good lower bound
is achieved after a small number of iterations in every variant, and the remaining part of the computations
is mostly to certify the answer.

A variant of our algorithm is that instead of $\eps$, we are given a time
budget and want to compute the best possible (relative) approximation
in this time limit. In such a case, we propose to traverse the quad-tree
by always subdividing the box with the largest upper bound 
(i.e., the output of the \texttt{Bound} primitive).
When the time is over, it suffices to peek at the top of the priority
queue to get the current upper bound, and we can output the lower bound $\rho$
and the relative error that we can guarantee at this moment.
It is instructive to plot how the relative error decreases as the algorithm runs;see \cref{fig:error-time}.
For instance, we can see that it takes approximately 3.5 times longer to bring the relative error below 0.1 than below 0.2,
if we use the constant bound.
This agrees with the complexity estimate in Theorem~\ref{thm:relative_approx_complexity}.

One detail in this plot is relevant for the experiments of the previous subparagraph.
If we choose a relative error $\eps_0$ and draw a horizontal line
$\eps = \eps_0$ in \cref{fig:error-time} until it intersects the plotted curves,
then the $x$-coordinate  of the intersection
is
the time that our algorithm needs to guarantee a {${1+\eps_0}$} approximation with the corresponding bound.
We can see that the difference between
the time needed with the global bound and the time needed with the constant bound
is not large for some values of $\eps_0$, but a small change of $\eps_0$
can rapidly increase it.  Clearly, this is highly input-specific, and this partially
explains the large variation in the improvement ratios that we observed above, when we ran experiments
with fixed $\eps$.

We provide additional experimental results in the \cref{sec:more_experiments}.
It contains the reduction rate (the measure used in \cite{italian}) for the experiments
of this section, scaling results on the dataset \dsR, and heatmap visualization of $d_B(F^1_\slice, F^2_\slice)$.

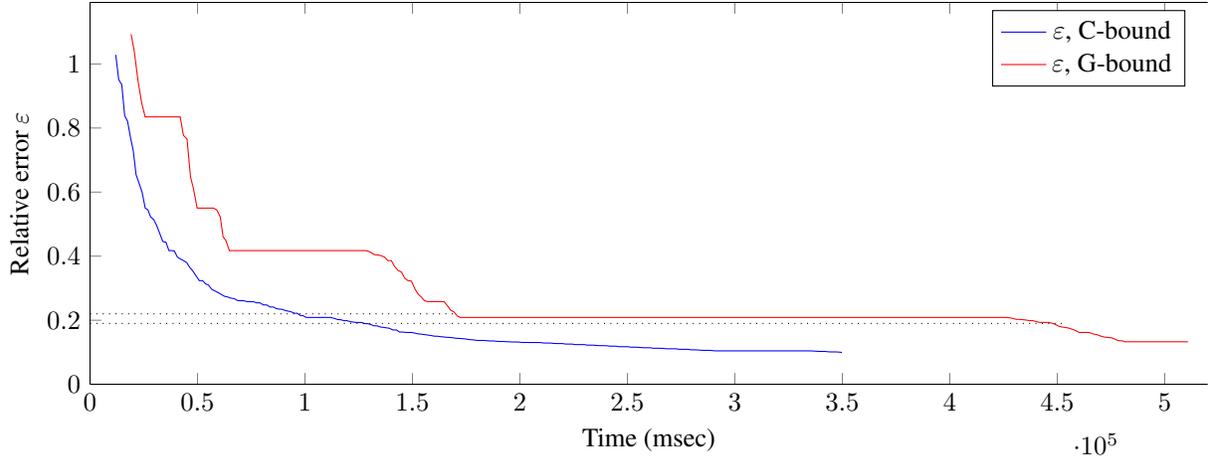
\begin{figure}[t]
\centering
\begin{tikzpicture}
\pgfplotsset{
    scale only axis,
    xmin=0, xmax=520000
}
\begin{axis}[ 
        axis y line* = left,
                height=2.00in,
                width=.9\textwidth,
                xlabel={Time (msec)},
                ytick distance = 0.2,
                ylabel={Relative error $\eps$},
            ]
    \addplot    +[mark=none] table[x=time, y=error, col sep = tab]  {data/ub_refined_second_cat0_david8_0.1.txt}; \label{plot_two}
    \addlegendentry{$\eps$, C-bound};
    \addplot    +[mark=none] table[x=time, y=error, col sep = tab]  {data/ub_grob_second_cat0_david8_0.1.txt}; \label{plot_three}
    \addlegendentry{$\eps$, G-bound};
    \addplot +[mark=none, black, dotted] coordinates {(0,0.19) (452095,0.19)};
    \addplot +[mark=none, black, dotted] coordinates {(0,0.22) (170095,0.22)};
\end{axis}
\end{tikzpicture}
\caption{Error decay with time for the C- and G-bounds.}
\label{fig:error-time}
\end{figure}

\section{Conclusion}
\label{sec:conclusion}
We presented an algorithm for the matching distance
that keeps subdividing boxes until a sufficiently close approximation
of the matching distance can be guaranteed.
This high-level description also applies to the previous
approach by Biasotti et al., which raises the question of how the approaches
compare in the details.
Instead of pointing out similarities and differences in the technical part,
we give a detailed discussion on this topic 
in Appendix~\ref{app:comparison}.

\medskip

We have restricted to the case of bi-filtrations in this work.
Generalizations in several directions are possible. First of all,
instead of bi-filtrations, our algorithm works the same
when the input is a pair of presentations of persistence modules~\cite{lw-computing,klo-exact,bbk-computing}.
Since a minimal presentation of a bi-filtration
can be of much smaller size than the bi-filtration itself, and its computation
is feasible~\cite{lw-computing,dx-generalized},
switching to a minimal presentation will most likely increase the
performance further. We plan to investigate this in further work.
Moreover, the case of $k$-critical bi-filtrations can be handled
with our methodology, just by defining the push of a simplex as the 
minimal push over all its critical values (Appendix~\ref{app:k_critical}).
Our approach can also be combined with \emph{barcode templates}
as introduced in the \rivet library~\cite{lw-interactive}.
Finally, an extension of our approach to $3$ and more parameters should be
possible in principle, but we point out that the space of affine lines
through $\R^d$ is $2(d-1)$ dimensional. Hence, already the next case
of tri-filtrations requires a subdivision in $\R^4$, and it is questionable
whether reasonably-sized instances could be handled by an extended
algorithm.

As the experiments show, in some cases the cost of computing the L-bound 
makes the variant with the {C-bound} faster.
However, parallelization of the L-bound
is trivial, and we believe that on larger instances it will make the L-bound the best choice.

Finally, since the matching distance can be computed exactly in polynomial
time~\cite{klo-exact}, the question is whether there is a practical
algorithm for this exact computation. Our current implementation can serve
as a base-line for a comparison between exact and approximate version
of matching distance computations that hopefully lead to further
improvements for the computation of the matching distance.

\bibliography{bib}

\begin{appendix}

\section{Proof of correctness}
\label{app:correctness}
We argue that the approximation algorithm given in Section~\ref{sec:approx_alg}
indeed yields a $\delta$ such that
\[
\dmatch(F^1,F^2)-\eps \leq \delta\leq \dmatch(F^1,F^2).
\]
Firstly, the value $\rho$ is non-decreasing during the algorithm.
Call a box \emph{terminal} if it does not get subdivided during
the algorithm (in other words, considering the set of boxes processed
in the algorithm as a \emph{quad-tree}, terminal boxes correspond
to leaves in the quad-tree). By construction, we have that
for each terminal box and each slice $\slice$ specified by the box
\[
d_B(F^1_\slice,F^2_\slice)\leq\mu\leq \rho+\eps\leq \delta+\eps
\]
Moreover, the terminal boxes form a cover of the initial boxes;
more precisely, the union of all terminal boxes with type $y$-steep
is equal to $[0,1]\times[0,Y]$, and similarly for the other three types.
Note that the initial boxes correspond to the sets $\slices_1,\ldots,\slices_4$
from (\ref{eqn:matching_compact}). Hence, writing $\mathcal{T}$
for the set of all slices contained in a terminal box,
we have that $\mathcal{T}=\slices_1\cup\ldots\cup\slices_4$ and it follows that
\begin{align*}
\dmatch(F^1,F^2)=&\sup_{\slice\in\slices_1\cup\ldots\cup\slices_4} d_B(F^1_\slice,F^2_\slice)\\
=&\sup_{\slice\in\mathcal{T}} d_B(F^1_\slice,F^2_\slice)\leq\delta+\eps
\end{align*}
proving the first inequality. On the other hand, since $\rho$
is always a bottleneck distance for some slice, we have that $\rho\leq \dmatch(F^1,F^2)$ throughout the algorithm, 
hence also $\delta\leq \dmatch(F^1,F^2)$.

\section{Detailed proof of Theorem~\ref{thm:corner-point}}
\label{app:proof_hor_vert}
To prove the theorem, it will be convenient to define
\[
D(\sslope,\smu):=|\wpush_p(\slice_{\sslope,\smu})-\wpush_p(\centerslice)|,
\]
where $\slice_{\sslope,\smu}$ is the slice corresponding to $(\sslope,\smu)$.
It follows immediately that
\[v(p,\slicebox)=\max_{(\sslope,\smu) \in \slicebox} D(\sslope,\smu).\]
and the theorem states that $D$ is maximized in one 
of the corners of $\slicebox$.
The next lemma reveals the structure of $D$ along vertical and horizontal
line segments within~$\slicebox$.

\begin{mylem}\label{lem:hor_vert}
Let $\slicebox=[a,b]\times[c,d]$. For any $\sslope_0\in[a,b]$, the function
\[
D(\sslope_0,\cdot): [c,d]\to\R
\]
is maximized at $c$ or at $d$. Likewise, for $\smu_0$ fixed, the function
\[
D(\cdot,\smu_0): [a,b]\to\R
\]
is maximized at $a$ or at $b$. 
\end{mylem}
\begin{proof}
\newcommand{\topslice}{\slice_{\text{top}}}
\newcommand{\bottomslice}{\slice_{\text{bottom}}}
Throughout the proof, we write $w:=\wpush_p(\centerslice)$,
which is a constant independent of $\sslope$ or $\smu$.
Let us consider the case of flat $y$-slices first. 
For $\sslope_0\in[a,b]$ fixed,
write $\topslice$ for the slice $(\sslope_0,d)$
and $\bottomslice$ for the slice $(\sslope_0,c)$.
There are three possible locations of the point $p$: it can be
above both $\bottomslice$ and $\topslice$, below both of them,
or above $\bottomslice$ and below $\topslice$.

In the first case, $D$ takes the form
\[D(\sslope_0,\smu)=|p_y-\smu-w|\]
on the whole interval $[c,d]$ (cf. \cref{tbl:wpush-formulas}).
This is a ``V-shaped'' function which takes its maximum at a boundary point.

In the second case, $D$ takes the form
\[D(\sslope_0,\smu)=|\sslope_0p_x -w|\]
which is a constant function, clearly also being maximal at either
boundary point.

In the third case, there is a unique point $\xi\in[c,d]$ such that $p$ lies
on the slice parameterized by $(\sslope_0,\xi)$. Then, on the interval
$[c,\xi]$, $D$ is a V-shaped function as above, and on $[\xi,d]$, $D$
is a constant function (and the two branches coincide at $\xi$, 
since $D$ is continuous). It follows that also in this case, $D(\sslope_0,\cdot)$
is maximized at a boundary point.

The analysis of the function $D(\cdot,\smu_0)$ for $\smu_0\in[c,d]$
is very similar. We write $\topslice$ for the slice $(b,\smu_0)$
and $\bottomslice$ for $(a,\smu_0)$. Note that the slices $(\sslope,\mu_0)$
for $\sslope\in [a,b]$ correspond to the slices obtained when rotating
from $\bottomslice$ to $\topslice$ in counterclockwise direction
with fixed origin $(0,\smu_0)$. There are same three possible cases
as above for the location of $p$ with respect to $\topslice$ and $\bottomslice$.

If $p$ is above both slices, $D$ takes the form
\[D(\sslope,\smu_0)=|p_y-\smu_0-w|\]
on $[a,b]$,
which is a constant function. If $p$ is below both slices, $D$ takes the
form
\[D(\sslope,\smu_0)=|\sslope p_x -w|\]
which is a V-shaped function.
If $p$ is in-between the slices, there is again a unique $\xi$ such that
$p$ lies on the slice $(\xi,\smu_0)$, and the function splits into
a V-shaped branch and a constant branch. This proves the statement for the
case of flat $y$-slices.

The other three cases are analogous: In all of them, the functions
$D(\sslope_0,\cdot)$ and $D(\cdot,\smu_0)$ are either $V$-shaped, constant,
or a combination of both.
\end{proof}

\begin{proof}(of Theorem~\ref{thm:corner-point})
Let $(\sslope,\smu)$ be any point in $\slicebox$. By Lemma~\ref{lem:hor_vert}
we can move $(\sslope,\smu)$ vertically to either the lower
or upper boundary without decreasing the $D$-value. Then, using
Lemma~\ref{lem:hor_vert}, we can move the point horizontally to
one the corners, again without decreasing the $D$-value.
The statement follows.
\end{proof}

\section{Detailed proof of Theorem~\ref{thm:local_constant_bound}}
\label{app:constant_bound}
\newcommand{\intermslice}{\tilde{\slice}}
\newcommand{\intermsslope}{\tilde{\sslope}}
\newcommand{\intermsmu}{\tilde{\smu}}
Let us start with the case of flat $y$-slices.
We first prove that for any two slices $\slice = (\sslope,\smu)$ and $\slice' = (\sslope',\smu')$
and any $p \in [0, X] \times [0, Y]$, it holds that
\begin{eqnarray*}
| \wpush_p(\slice) -\wpush_p(\slice') | \leq |\smu - \smu'| + X |\sslope - \sslope'|
\end{eqnarray*}
We consider three cases: if $p$ is above both $\slice$ and $\slice'$,
\[
 | \wpush_p(\slice)-\wpush_p(\slice') | = |(p_y + \smu)-(p_y + \smu')|=|\smu - \smu'|,
\]
and the bound clearly holds.
If $p$ is below both slices,
\begin{eqnarray*}
    | \wpush_p(\slice) -\wpush_p(\slice') | = |\sslope p_x - \sslope' p_x | \leq p_x | \sslope - \sslope'|,
\end{eqnarray*}
and the bound holds, because $X$ is the maximal possible value of $p_x$.
Finally, if $p$ is above $\slice$ and below $\slice'$ (or vice versa),
then the line segment connecting $(\sslope, \smu)$ and $(\sslope', \smu')$ 
contains at least one point $(\intermsslope, \intermsmu)$ such that the
slice $\intermslice$ defined by these parameter values contains the point $p$.
Since $p$ is on $\intermslice$, we can use either formula for the weighted
push in the corresponding column of \cref{tbl:wpush-formulas} for $\intermslice$.
Together with the triangle inequality, we obtain:
\begin{align*}
    & | \wpush_p(\slice) -\wpush_p(\slice') |  \\
   = & | \wpush_p(\slice) - \wpush_p(\intermslice) |  + | \wpush_p(\intermslice) - \wpush_p(\slice')|\\
   = & | (p_y - \smu) - (p_y - \intermsmu)| + | \intermsslope p_x - \sslope p_x |\\
   \leq  &  | \intermsmu - \smu |  + p_x | \intermsslope  - \sslope' | \\
   \leq & | \smu'- \smu | + X | \sslope'- \sslope |,
\end{align*}
    where the last inequality holds, because $\intermsmu$ is between $\smu$ and $\smu'$, and $\intermsslope$ is
    between $\sslope$ and $\sslope'$.

The first case of the statement follows at once by setting $L'$
to be the center of the box $\slicebox$, and $L$ to be any point in $\slicebox$
since in this case, $|\sslope' - \sslope| \leq \Delta \sslope / 2$
    and $|\smu' - \smu| \leq \Delta \smu / 2$.

Next we consider the case of steep $y$-slices.
We claim that for any two slices $\slice = (\sslope,\smu)$ and $\slice' = (\sslope',\smu')$
and any $p \in [0, X] \times [0, Y]$, it holds that
\begin{align*}
  & | \wpush_p(\slice) -\wpush_p(\slice') | \\
\leq & \sslope  |\smu - \smu'| + ( Y - \smu') |\sslope - \sslope'|
\end{align*}
We consider the same three cases as above. If $p$ is below both $\slice$
and $\slice'$, the weighted pushes are both equal to $p_x$, 
and the difference is 0. 
If $p$ is above both $\slice$ and $\slice'$, we calculate
\begin{align*}
    & | \wpush_p(\slice) -\wpush_p(\slice') | \\
  = & | \sslope (p_y - \smu) - \sslope' (p_y - \smu') | \\
  =  & | \sslope (p_y - \smu) - \sslope (p_y - \smu')  + \sslope (p_y - \smu')  - \sslope' (p_y - \smu') | \\
 \leq & | \sslope (\smu' - \smu)| + |(\sslope-\sslope')(p_y-\smu')|\\
 = & \sslope |\smu'-\smu| + (p_y-\smu')|\sslope-\sslope'|.
\end{align*}
Note that in the last line, we use that $p_y-\smu'$ is positive, which
follows from $p$ being above $\slice'$, and $p_x\geq 0$.

In the third case, $p$ is above $\slice$ and below $\slice'$.
Instead of the line segment connecting them, 
we consider the path from $(\sslope,\smu)$ to $(\sslope',\smu')$
that first goes vertically to $(\sslope,\smu')$, and then horizontally
to $(\sslope',\smu')$. Also on this path, there is a slice 
$\intermslice$ such that $p$ is on $\intermslice$. 
Then, by triangle inequality,
the difference of the weighted pushes for $\slice$ and $\slice'$
is at most
\[
| \wpush_p(\slice) - \wpush_p(\intermslice) |  + | \wpush_p(\intermslice) - \wpush_p(\slice')|
\]
and the second term is equal to $0$. Hence, using the previous calculation,
the difference can be bounded by
\[
\sslope |\intermsmu-\smu| + (p_y-\intermsmu)|\sslope-\intermsslope|
\]
Now, if $\intermslice$ is on the vertical branch of the path, $\sslope=\intermsslope$, and the second term vanishes. If $\intermslice$ is on the horizontal
part, $\intermsmu=\smu'$, and the bound above is equal to
\[\sslope |\smu'-\smu| + (p_y-\smu')|\sslope-\intermsslope|\]
and the bound follows because $|\sslope-\intermsslope|\leq|\sslope-\sslope'|$.

Using this estimate on $| \wpush_p(\slice) - \wpush_p(\slice') |$,
the theorem statement for steep $y$-slices 
follows by choosing $\slice$ as the center slice, and $\slice'$
as any other slice in $\slicebox$. Note that using choosing $\slice'$
as the center slice instead yields the (seemingly) different bound
\[
\frac{1}{2}\left(\sslope_{\max} \Delta \smu + (Y - \centersmu) \Delta \sslope \right)\},
\]
but it can be verified by a simple calculation that both bounds are equal.
The bounds for $x$-slices are proved analogously.

\section{Complexity proofs}
\label{app:complexity}

\begin{mylem}\label{lem:k_bound}
A level-$k$-box with 
\[k:=\left\lceil\log\frac{2C}{\eps}\right\rceil\]
(where the logarithm is with base $2$)
does not get subdivided by the algorithm.
\end{mylem}
\begin{proof}
We have to show that $\mu\leq\rho+\eps$, where $\mu$ is the upper bound
computed by the \texttt{Bound} primitive, and $\rho$ is largest 
weighted bottleneck distance encountered at the moment 
when the algorithm decide whether to subdivide $\slicebox$.
Note that in this moment, $\rho\geq d_B(F^1_{\centerslice},F^2_{\centerslice})$
is ensured because the latter value has been computed
in the previous step. Hence, using the previous lemma,
\[
\mu\leq d_B(F^1_{\centerslice},F^2_{\centerslice})+2C2^{-k}\leq \rho+2C2^{-k}\leq \rho+\eps,
\]
where the last step follows from the choice of $k$.
\end{proof}
\begin{proof}(of Theorem~\ref{thm:absolute_approx_complexity})
With $k$ as above, the worst case is that the algorithm
considers all boxes of level $k$. In that case, the total number of boxes
considered is 
\[4(4^0+4^1+\ldots+4^k=O(4^k)\]
Plugging in $k$ yields that $4^k=(2^k)^2=O(\left(\frac{C}{\eps}\right)^2)$
as the number of considered boxes. On each box, the algorithm evaluates
the weighted bottleneck distance at the center slice, which requires
the computation of weighted pushes, of two persistence diagrams,
and of their bottleneck distance.
Complexity-wise, the dominating step is the persistence computation,
which we do in $O(n^3)$ steps (this complexity can be reduced
to $O(n^\omega)$, where $\omega$ is the matrix multiplication constant~\cite{mms-zigzag}).
The complexity bound follows.
\end{proof}

\subparagraph{Relative approximation}
For the proof of Theorem~\ref{thm:relative_approx_complexity}, we need the following lemma stating that 
$\rho$ will eventually be a close approximation of the matching distance.
\begin{mylem}
\label{lem:rho_close_to_dmatch}
Write $\dmatch:=\dmatch(F^1,F^2)$ and assume $\dmatch>0$.
For
\[
k\geq 1+\left\lceil\log\frac{(1+\eps)2C}{\eps \dmatch}\right\rceil,
\]
it holds that when the algorithm considers a level-$k$-box,
we have that $\rho\geq\dmatch/(1+\eps)$.
\end{mylem}
\begin{proof}
\newcommand{\optimalslice}{\slice^\ast}
Let $\optimalslice$ be a slice such that the matching distance is realized
as $d_B(F^1_{\optimalslice},F^2_{\optimalslice})$. Note that the algorithm
handles all boxes of level $<k$ before handling any box of level $k$.
Now, we distinguish two cases:

The first case is that $\optimalslice$ lies in some box  
of level $<k$ for which the algorithm did not subdivide further.
That means that $\mu\leq (1+\eps)\rho'$, where $\mu$ is an upper bound
for the box and $\rho'$ is the value of $\rho$ at this moment of the 
algorithm. Note that since $\optimalslice$ lies in $\slicebox$,
$\mu\geq\dmatch$ must hold. Also, $\rho'\leq\rho$ because $\rho$
only increases. It follows that
\[
\dmatch\leq\mu\leq(1+\eps)\rho'\leq(1+\eps)\rho
\]
proving the statement for the first case.

The second case is that $\optimalslice$ lies in some level-$(k-1)$-box
$\slicebox$ which has been subdivided. 
By the way how $\mu$ is computed, we have that 
\[\mu\leq d_B(F^1_{\centerslice},F^2_{\centerslice})+2C2^{-(k-1)},\]
where $\centerslice$ is the center slice of the box. Moreover, as before,
$\dmatch\leq\mu$ holds, and $d_B(F^1_{\centerslice},F^2_{\centerslice})\leq\rho$
because $\rho$ is updated using $\centerslice$. In summary, we obtain
\[
\dmatch\leq\rho+2C2^{-(k-1)}
\]
The bound on $k$ ensures that 
\[
2C2^{-(k-1)}\leq 2C 2^{\log\frac{\eps\dmatch}{2C(1+\eps)}}=\frac{\eps\dmatch}{1+\eps}
\]
so we obtain
\[
\rho\geq \dmatch-2C2^{-(k-1)}\geq \dmatch - \frac{\eps\dmatch}{1+\eps} =\frac{\dmatch}{1+\eps}.
\]
\end{proof}

\begin{mylem}
If $\dmatch>0$, 
a level-$k$-box with 
\[k:=1+\left\lceil\log\frac{(1+\eps)2C}{\eps\dmatch}\right\rceil\]
does not get subdivided by the relative approximation algorithm.
\end{mylem}
\begin{proof}
Let $\slicebox$ be some level-$k$-box. We have to show that
$\mu\leq (1+\eps)\rho$ for $\slicebox$. Note that 
\[
\mu\leq d_B(F^1_{\centerslice},F^2_{\centerslice})+2C2^{-k}\leq\rho+\frac{\eps\dmatch}{(1+\eps)}
\]
Moreover, the $k$ in question satisfies the assumptions of Lemma~\ref{lem:rho_close_to_dmatch}, so we have that $\dmatch\leq(1+\eps)\rho$, and we obtain
\[
\mu\leq \rho+\frac{\eps\dmatch}{(1+\eps)}\leq (1+\eps)\rho
\]
\end{proof}

\begin{proof}(of Theorem~\ref{thm:relative_approx_complexity})
The proof is the analogous to the proof of Theorem~\ref{thm:absolute_approx_complexity}, noting that the algorithm has to consider 
\[O(4^k)=O(\left( \frac{(1+\eps)2C}{\eps\dmatch}\right)^2)\]
boxes, and the cost for each box is $O(n^3)$.
\end{proof}

\section{Comparison with~\cite{italian}}
\label{app:comparison}
Many ideas used in our paper appear in~\cite{italian} in some form. 
For instance,
the idea of restricting the parameter space to a bounded region 
(see the discussion after our Lemma~\ref{lem:wpush}) corresponds
to Lemma~3.1 in~\cite{italian}. However, instead of completely disregarding
the region outside of the bounded region, they introduce two points
in this domain, which seems unnecessary and complicates the algorithmic
description. As another example, our upper bound in Theorem~\ref{thm:local_constant_bound} corresponds 
to their bound from Lemma~3.3 (with worse constants
which we discuss below). The proof of their lemma contains a case distinction,
where they consider, expressed in our notation, 
the case of two flat slices, two steep slices,
and the mixed case of a flat and a steep slice.
By our choice of splitting the parameter space in $4$ parts, we can ensure
that all slices within a box are of the same type, which makes the (tedious)
analysis of mixed cases unnecessary.

Moreover, we improve on~\cite{italian} in several algorithmic ways: foremost,
our local bounds provide better estimates of the variation and lead to
fewer subdivisions. Moreover, when we subdivide boxes, we keep the
aspect ratio of the box the same in the next iteration. This allows us
to cover the initial box with $4^k$ boxes of level $k$.
The approach in~\cite{italian} uses squares instead and hence requires
$C4^{k}$ boxes on level $k$. Since their algorithm has to subdivide
to a level of $O(C/\eps)$ in the worst case (same as ours), the complexity
becomes $O(n^3\frac{C^3}{\eps^2})$, which is a factor $C$ worse.
Finally, the approaches differ in the choice of the parameterization.
Their approach, restated in geometric terms, represents a slice $\slice$ by
two parameters $(\lambda,\beta)$, where $\lambda$ is the sine of the angle
of the $\slice$ with the $x$-axis, 
and the origin is chosen as the point $(\beta,-\beta)$.
While this approach has the pleasant effect of avoiding a case distinction
between $x$-slices and $y$-slices, it has two downsides: first of all, 
the bounding rectangle in parameter space contains more slices than
in our version, and the fact that the origin is further away from the critical
points (which all lie in the upper-right quadrant) leads to a worsening
of the bounds in Theorem~\ref{thm:local_constant_bound}.
This partially explains the discrepancy of their upper bound of $(16C+2)\delta$
(Theorem 3.4 in~\cite{italian}) and the bound of $4C\delta$ that could be
achieved with our methods.

\section{Datasets used in experiments}
\label{sec:datasets}
\subparagraph{Lower-star bi-filtrations.}
A simple way to construct bi-filtration of a simplicial complex $K$
is the \emph{lower star filtration}. For that, assume that $\ffunc(\cdot)=(\ffunc_1(\cdot),\ffunc_2(\cdot))$
is defined on the vertices of $K$ (e.g., consider the case that $K$
is a triangulated mesh, and $\phi$ measures properties at the vertices
such as distance to the barycenter or local discrete curvature).
Now set, for a $d$-simplex $\sigma=\{v_0, \dots, v_{d}\}$,
\[
\ffunc(\sigma):=(\max_{i=0,\ldots,d}\ffunc_1(v_i),\max_{i=0,\ldots,d}\ffunc_2(v_i)).
\]
Geometrically, $\ffunc(\sigma)$ is the smallest point $q$ in $\R^2$ with respect
to $\leq$ such that the lower-right quadrant with center $q$ (as in Figure~\ref{fig:bi-filtration} (middle)) contains $\ffunc(v_0),\ldots,\ffunc(v_d)$.

Both datasets \dsGH and \dsED are based on a Non-Rigid World Benchmark \cite{nrwb},
a collection of 3D-shapes represented as triangular meshes.
A triangular mesh is a geometric realization of a 2-dimensional
simplicial complex, and, if we fix a function $\phi \colon V \to \RR^2$
on the vertices of a mesh, this gives rise to a lower star filtration.

For \dsGH we use the function $\phi^{GH} = (\phi_1^{GH} / K_1, \phi_2^{GH} / K_2)$ defined as follows:
\begin{align}
    \phi_1^{GH}(v) &= \mbox{integral geodesic distance of $v$};\\
    \phi_2^{GH}(v) &= \mbox{HKS at $t = 1000$;} \\
    K_1 & = \max_v \phi^{GH}_1(v); \\
    K_2 & = \max_v \phi^{GH}_2(v).
\end{align}

HKS in the formula for $\phi_2^{GH}$ stands for Heat Kernel Signature. It is computed by solving the discrete analogue
of the heat equation on smooth manifolds (where the discrete Laplacian replaces the Laplace-Beltrami operator). 
The heat kernel was introduced in \cite{heat-kernel}, and became a very popular tool in shape analysis;
we used the publicly available code from \cite{pyhks-repo} to compute it.

Let $s(v, w)$ be the length of the shortest path that connects $v$ and $w$ on the mesh and does not go through
the interior of any of the mesh triangles. The integral geodesic distance $\phi_1^{GH}(v)$ is
a weighted sum over vertices $w \neq v$ of $s(v,w)$;
it was introduced in \cite{hilaga2001topology}, and we refer the reader to this paper
for further details and motivation. We compute $s(v,w)$ using Dijkstra's algorithm
on the graph that consists of the mesh vertices and edges with weight of an edge being Euclidean distance between its endpoints.

Normalization constants $K_1, K_2$ ensure that the maximal coordinates $X$ and $Y$ are 1.

Let us now describe the dataset \dsED.
We fix a mesh with vertices $V = \{v_1, \dots, v_n\}$, and let $b$ be the center of mass of the mesh, $b = \frac{1}{n} \sum_{i = 1}^n v_i$.
Vector $\vw$ is defined as
\[
    \vw = \frac{ \sum_{i = 1}^n (v_i - b ) \| v_i - b \| }{ \sum_{i = 1}^n \| v_i - b \|^2}.
\]
Let $\pi_{\vw}$ be the plane passing through $w$ and orthogonal to $\vw$,
$\ell_{\vw}$ be the line passing through $w$ and parallel to $\vw$;
$d(v, \ell)$ and $d(v, \pi)$ denote Euclidean distances from the point $v$ to the line $\ell$
and plane $\pi$. The function $\phi^{ED} = (\phi_1^{ED} / K_1, \phi_2^{ED} / K_2) $ that defines the $\dsED$ bi-filtrations is given by
\begin{align}
    \phi_1^{ED}(v_i) &= 1 - \frac{ d(v_i, \ell_{\vw}) } { \max_{k=1,\dots,n} {d(v_k, \pi_{\vw})} }; \\
    \phi_2^{ED}(v_i) &= 1 - \frac{ d(v_i, \pi_{\vw}) } { \max_{k=1,\dots,n} {d(v_k, \pi_{\vw})} }; \\
    K_1 & = \max_v \phi^{ED}_1(v); \\
    K_2 & = \max_v \phi^{ED}_2(v).
\end{align}
These formulas can be found on page 1743, Section 4.3 of \cite{italian}.
The \dsED dataset can be easily computed, since the formulas involve only elementary geometric calculations.

We generated 70 bi-filtrations that cover different classes of the benchmark (male and female figures 
in different poses, seahorses, cats, etc), hence we have 2,415 pairs to test.

Note that \dsGH and \dsED are different from \cite{italian}, because
the authors also applied additional transformations to the meshes
before computing the bi-filtrations; we skipped this step, because it is not clearly
described and hard to reproduce.

\subparagraph{Random bi-filtrations.}
In order to see the scaling on larger inputs, we generated random bi-filtrations as follows.
The input parameters are the number of maximal simplices $M$,
total number of vertices $N$, and the dimension of maximal simplices $D$.
We randomly generated $M$ distinct subsets of cardinality $D+1$
of the set $\{1, \dots, N\}$, and for each of these maximal simplices $\sigma$
we chose the $x$ and $y$ coordinates of $\ffunc(\sigma)$ uniformly at random from $[0, 1000]$.
Assume now that we defined $\ffunc$ on all simplices in some dimension $d+1 \leq D$,
and we want to define $\ffunc(\tau)$ on a simplex $\tau$ of dimension $d$.
We know that $\ffunc(\tau)$ must appear in the bi-filtration
before any of its co-faces.
To ensure that, we intersect all rectangles with the bottom left corner at $(0,0)$
and the upper right corner at $\ffunc(\sigma)$, where $\sigma$ ranges over all simplices such that $\tau$ is a face of $\sigma$
and $\ffunc(\sigma)$ is defined. This intersection is itself a rectangle; we pick
an integral point in this rectangle uniformly at random as $\ffunc(\sigma)$.
We generated 6 random bi-filtrations of dimension 1 for each $N \in \{500, 1000, 2000\}$,
with  $M = 4N$. These bi-filtrations form the dataset \dsR.

\section{Additional Experimental Results}
\label{sec:more_experiments}
\subparagraph{Reduction rate.}
The only measure reported in \cite{italian} is what the authors called
the \textit{reduction rate}. It is defined as $1 - c / 4^k$, where
$c$ is the number of calls, $k$ is the level of the deepest box in the quad-tree, on whose
center the algorithm actually performs a call (i.e., computes the weighted
bottleneck distance), and $4^k$ is the total number of boxes on level $k$.
What does the reduction rate measure? Suppose that for some reason we decided to look at the $k$-th level of 
the quad-tree only. We can simply compute the weighted bottleneck distance
at the center of each of the $4^k$ boxes and output the largest result; this would be
a brute force approach. If we have a bound and guess the level $k$ correctly,
then we can guarantee the desired approximation quality.
The reduction rate shows which fraction of the $4^k$ calls
we avoid by switching to the quad-tree algorithm; if it is 0.99, this means
that we avoided 99\% of calls.

We give our reduction rates (average, minimal, and maximal) in \cref{tbl:red_rate}.
The best reported reduction rate in \cite{italian} for $\eps = 0.1$ is 94.7\% (\dsGH) and 93.8\% (\dsED),
with average values of 60.6\% and 57\%; minimal values are not provided.
Our maximal values for the global bound approximately
agree with theirs,\footnote{The maxima need not be the same, since our datasets are not exact copies of theirs.}
and our average values are by roughly 10\% higher.

This means that with the G-bound we can on average avoid not 57\% of calls, as in \cite{italian},  
but 71\%. This should be expected, because the G-bound has a smaller constant factor,
see the end of \cref{app:comparison}. The advantage of the local and constant bounds
becomes evident in the worst case; we never have reduction by less than 86\% with the local bound,
while the global bound can go as low as 30\%.

\begin{table}\centering
    \ra{1.2}
\begin{tabular}{@{}rrrrrrr@{}}\toprule
    & \multicolumn{3}{c}{Global} & \multicolumn{3}{c}{Local Constant} \\
                          \cmidrule{2-4} \cmidrule{5-7} 
                     & Max & Avg   & Min  & Max  &   Avg  & Min  \\ \midrule
\dsGH & 92.95 & 69.47 & 56.43 & 96.95 & 86.97 & 75.77  \\
 \dsR & 64.44 & 56.28 & 32.88 & 96.64 & 91.23 & 81.86  \\
\dsED & 97.45 & 71.00 & 30.78 & 98.01 & 89.07 & 69.35  \\
      & \multicolumn{3}{c}{Local Linear} \\
        \cmidrule{2-4} 
      & Max & Avg   & Min \\ \midrule
\dsGH & 98.52 & 93.53 & 86.47 \\
\dsR     & 98.33 & 95.79 & 91.25 \\
\dsED & 99.30 & 96.49 & 86.39 \\
\bottomrule
\end{tabular}
    \caption{Reduction rate for $\eps = 0.1$, in percents.}
\label{tbl:red_rate}
\end{table}

\subparagraph{Scaling on random bi-filtrations.}
In datasets \dsGH and \dsED, the number of vertices is around 3,000. However,
since these datasets are lower-star bi-filtrations, the cardinality
of the persistence diagrams of the restrictions is smaller,
typically between 20 and 60 points. In dataset \dsR, the cardinality
of the diagrams in dimension 0 is almost equal to the number of vertices,
so it is larger by an order of magnitude even for 500
vertices. We use this dataset to show how our algorithm scales
on larger inputs.

\begin{table}\centering
    \ra{1.2}
\begin{tabular}{@{}rrrrrrr@{}}\toprule
    & \multicolumn{3}{c}{\#Calls} & \multicolumn{3}{c}{Time (sec)} \\
                          \cmidrule{2-4}  \cmidrule{5-7} 
\#vert. & L & C & G & L & C & G  \\ \midrule
   500 &    164 &    435 &    655 &    6.5 &      15.1  &     23.3 \\ 
  1000 &    893 &   2354 &   4222 &  197.2 &     505.7  &    983.5 \\ 
  2000 &   1233 &   3256 &   5995 &  911.8 &    2277.8  &   4254.3 \\ 
 \bottomrule
\end{tabular}
    \caption{Average number of calls and average running time with bounds L, C, and G for random bi-filtrations
    with different number of vertices. Relative error $\eps = 0.5$.}
\label{tbl:avg_calls_time_random}
\end{table}

In \cref{tbl:avg_calls_time_random}, we see how 
the running time and the number of calls grows as we increase
the size of the input. As in \cref{tbl:avg_calls_time},
the variance hidden behind the averaged numbers is large,
but the bounds compare in the expected way: the linear bound outperforms the constant bound,
and the constant bound outperforms the global bound.

In \cref{tbl:random_time_ratio} and \cref{tbl:random_call_ratio},
we provide the ratios of the running time and the number of calls.
These tables are similar to \cref{tbl:time_ratio} and \cref{tbl:call_ratio}.
For example, the first 3 columns of \cref{tbl:random_time_ratio} were obtained as follows.
For each pair of random bi-filtrations
with the same number of vertices, we measure the time
with the global bound and with the constant bound, and take their ratio.
The average of these ratios is in the first column, the minimum is in the second column,
and the maximum is in the third one.

Note that the ratios in these tables do not change much as we go from 500
vertices to 2,000.  If we compare the constant bound and the linear bound,
we notice that the ratios in \cref{tbl:random_call_ratio}
are very stable. The linear bound always reduces the number of calls
by a factor of approximately 2.5. The running time ratios in \cref{tbl:random_time_ratio}
increase for larger inputs, getting closer to the ratios of the number of calls.
This is actually expected, since the complexity of the \texttt{Eval} primitive
is super-linear, so the time spent in \texttt{Eval} starts to subsume the time spent on computing the L bound.
One should not expect that the ratio will grow with $n$: a better bound
only reduces the number of calls of \texttt{Eval}, but cannot accelerate
the computation of \texttt{Eval} itself.

\begin{table}\centering
    \ra{1.2}
\begin{tabular}{@{}rrrrrrr@{}}\toprule
                          & \multicolumn{3}{c}{Time: G / C} & \multicolumn{3}{c}{Time: C / L}\\
                          \cmidrule{2-7}
\#vert. & Avg & Min & Max        & Avg & Min & Max  \\ \midrule

   500 &   1.66 &   1.42 &   2.37 &   2.38 &   1.92 &   2.81 \\ 
  1000 &   1.76 &   1.37 &   2.22 &   2.53 &   2.34 &   2.77 \\
  2000 &   1.76 &   1.38 &   1.96 &   2.63 &   2.33 &   2.83 \\
\bottomrule
\end{tabular}
    \caption{Scaling on dataset \dsR, running time ratios for bounds G, C and L. Relative error $\eps = 0.5$}
\label{tbl:random_time_ratio}
\end{table}

\begin{table}\centering
    \ra{1.2}
\begin{tabular}{@{}rrrrrrr@{}}\toprule
                          & \multicolumn{3}{c}{\#Calls: G / C} & \multicolumn{3}{c}{\#Calls: C / L}\\
                          \cmidrule{2-7}
\#vert. & Avg & Min & Max        & Avg & Min & Max  \\ \midrule

   500 &   1.60 &   1.28 &   2.22 &   2.67 &   2.52 &   2.78   \\ 
  1000 &   1.68 &   1.25 &   2.03 &   2.65 &   2.57 &   2.80   \\ 
  2000 &   1.76 &   1.34 &   2.02 &   2.61 &   2.48 &   2.74   \\ 
\bottomrule
\end{tabular}
    \caption{Scaling on dataset \dsR, \#calls ratios for bounds G, C and L. Relative error $\eps = 0.5$}
\label{tbl:random_call_ratio}
\end{table}

\subparagraph{Heatmaps.}

Recall that we have 4 rectangles $\slices_1, \dots, \slices_4$ that parameterize 4 different types of slices.
Each slice passing through the origin is at the same time an $x$-slice and a $y$-slice,
and each slice with slope 1 is both flat and steep. Therefore
we can glue the rectangles $\slices_i$ together by identifying
points that represent the same slice, and we get a single domain $\slices$ with the slice
$y = x$ in the center, as in \cref{fig:slice_box_subdivision} (the coordinates in this figure do not
agree with the coordinates we use inside each $\slices_i$ for parameterization in \cref{sec:approx_alg}).  We can visualize values
of the weighted bottleneck distance by computing it at the center of each box of the quad-tree (of $\slices$) on a fixed level,
and using these values for the heatmap. The brighter a pixel of the heatmap,
the larger the value of $d_B(F^1_{\slice}, F^2_{\slice})$ for the corresponding slice $\slice$.

\begin{figure}
    \includegraphics[width=2.5in]{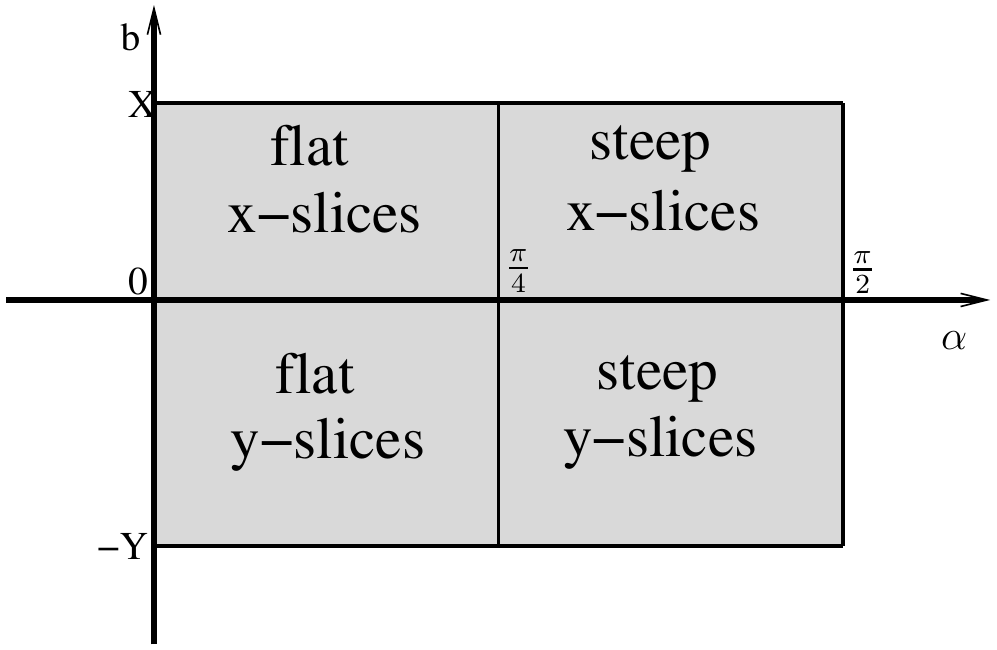}
    \caption{Decomposition of the domain $\slices$ into 4 types.}
    \label{fig:slice_box_subdivision}
\end{figure}

In \cref{fig:heatmap_good} and \cref{fig:heatmap_bad}, we show two examples (they were computed on the \dsED dataset)
of heatmaps.
For pair A, in \cref{fig:heatmap_good}, high values of the weighted bottleneck distance are concentrated
in a small bright spot around the center. 
We can expect that, whichever bound we use, the algorithm will need only a few subdivisions
in the darker area to ensure that these boxes cannot improve the lower bound.
The opposite is true for the heatmap of pair B, \cref{fig:heatmap_bad},
where a large part of all 4 quadrants has almost equal high values.
The algorithm will need to subdivide the boxes covering this part
until they become so small that the upper bound on each of them
will be within the error threshold. 
These expectations are confirmed by the experimental results.
It takes 325 calls to approximate the matching distance with $\eps = 0.1$
for pair A, and 838 calls for pair B,
if we use the linear bound in both cases.

Pair B is also an example of a case when the local constant bound does not
improve the performance in comparison with the global bound.
Indeed, since for pair B we have $X = Y = 1$, the global bound and the local constant bound agree
on two of the four quadrants (flat $y$-slices and steep $x$-slices, see \cref{eqn:glob_bound}), 
hence there will be no difference in the algorithm's
behavior there, and in this example all four quadrants require many subdivisions.

\begin{figure}[ht]
    \begin{centering}
        \begin{minipage}{0.45\textwidth}
            \includegraphics[scale=0.25]{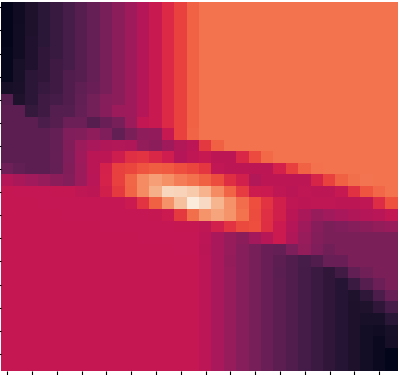}
            \caption{An example of a heatmap: pair A.}
            \label{fig:heatmap_good}
        \end{minipage}
        \hfill
        \begin{minipage}{0.45\textwidth}
            \includegraphics[scale=0.25]{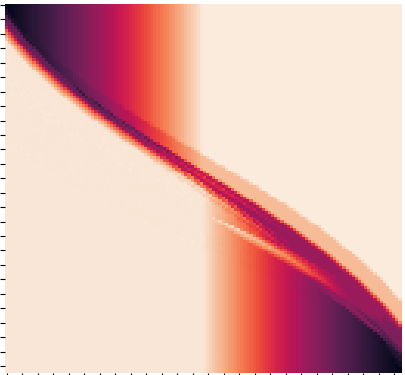}
            \caption{An example of a heatmap: pair B.}
            \label{fig:heatmap_bad}
        \end{minipage}
    \end{centering}
\end{figure}

\subparagraph{Additional plots.}
In \cref{fig:error-c-l-calls} we plot
the dependence of relative error on the number of calls
for the local constant and local linear bounds. The plot is for the same input
as \cref{fig:error-time}. Note that the plots are closer to each other
than the plots for the global and the local constant bound in \cref{fig:error-time},
and, if we chose to use time as the $x$-axis, the difference would become even smaller.

In \cref{fig:lower_upper_bound_calls} we plot
the evolution of the lower bound $\rho$ and the upper bound $\mu$
(for the same input). Note that the lower bound stabilizes very early.

\begin{figure}[t]
\centering
\begin{tikzpicture}
\pgfplotsset{
    scale only axis,
    xmin=0, xmax=5000
}
\begin{axis}[ 
        axis y line* = left,
                height=2.00in,
                width=.89\textwidth,
                xlabel={\#calls},
                ytick distance = 0.2,
                ylabel={Relative error $\eps$},
            ]
    \addplot    +[mark=none] table[x=calls, y=error, col sep = tab]  {data/ub_combined_second_cat0_david8_0.1.txt}; \label{plot_one_c_l}
    \addlegendentry{$\eps$, L-bound};
    \addplot    +[mark=none] table[x=calls, y=error, col sep = tab]  {data/ub_refined_second_cat0_david8_0.1.txt}; \label{plot_two_c_l}
    \addlegendentry{$\eps$, C-bound};
\end{axis}
\end{tikzpicture}
\caption{Error decay with time for C- and L-bounds.}
\label{fig:error-c-l-calls}
\end{figure}
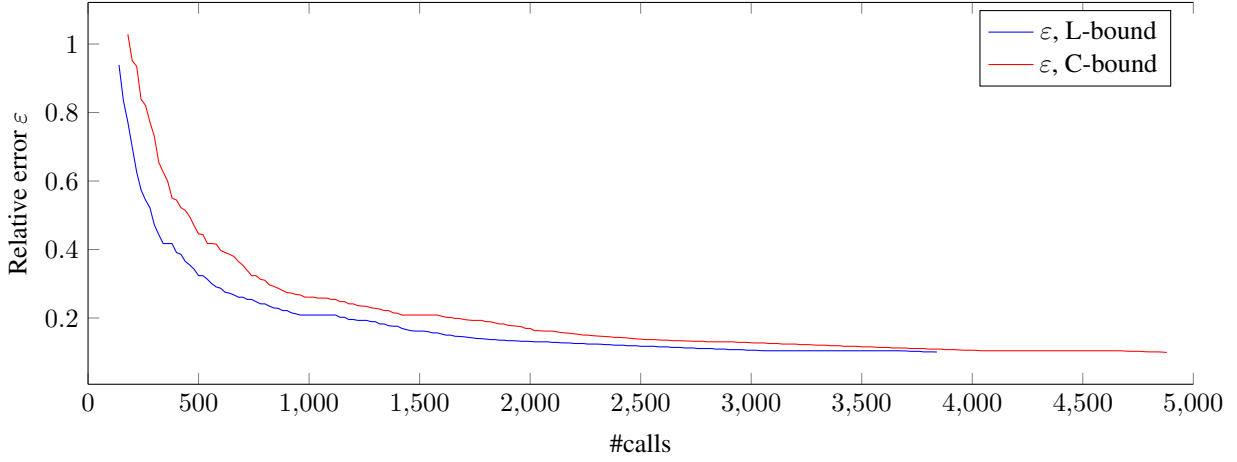

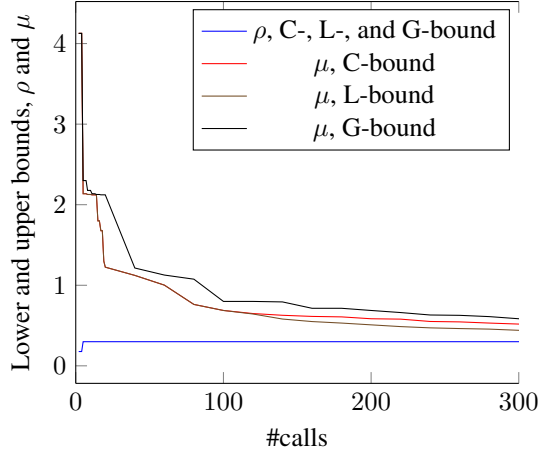
\begin{figure}[t]
\centering
\begin{tikzpicture}
\pgfplotsset{
    scale only axis,
    xmin=0, xmax=4000
}
\begin{axis}[
                axis y line* = left,
                height=2.00in,
                xlabel={\#calls},
                ylabel={Lower and upper bounds, $\rho$ and $\mu$},
                xmin=0, xmax=300
            ]
    \addplot    +[mark=none] table[x=calls, y=lower_bound, col sep = tab]  {data/ub_refined_second_cat0_david8_0.1_all.txt};
    \addplot    +[mark=none] table[x=calls, y=upper_bound, col sep = tab]  {data/ub_refined_second_cat0_david8_0.1_all.txt};
    \addplot    +[mark=none] table[x=calls, y=upper_bound, col sep = tab]  {data/ub_combined_second_cat0_david8_0.1_all.txt};
    \addplot    +[mark=none] table[x=calls, y=upper_bound, col sep = tab]  {data/ub_grob_second_cat0_david8_0.1_all.txt};
    \addlegendentry{$\rho$, C-, L-, and G-bound};
    \addlegendentry{$\mu$, C-bound};
    \addlegendentry{$\mu$, L-bound};
    \addlegendentry{$\mu$, G-bound};
\end{axis}
\end{tikzpicture}
    \caption{Lower and upper bounds.}
\label{fig:lower_upper_bound_calls}
\end{figure}

\section{Multi-critical filtrations.}
\label{app:k_critical}
Let $K$ be a simplicial complex. Let $\phi$ be a function
that assigns to every simplex $\sigma \in K$
a non-empty finite set $\phi(\sigma) = \{ p_1, \dots, p_s \}$ of points in $\RR^2$.
Each of the points $p_i$ is called a critical value of $\sigma$.
The cardinality $s$ of the set $\phi(\sigma)$ can depend on $\sigma$,
but we suppress this in the notation.
For $p \in \RR^2$, put
\[
    K_p := \{ \sigma \in K \mid p_i \pleq p \mbox{ for at least one } p_i \in \phi(\sigma) \}.
\]
In other words, the set $P_{\sigma} =  \{ p \in \RR^2 \mid \sigma \in K_p \}$
(the set of points $p$ such that $\sigma$ is present in $K_p$)
is bounded by a 'staircase' with corners at the critical values $p_1, \dots, p_k$ (we assume
that $p_i$ and $p_j$ are incomparable for $i \neq j$).
The pair $F = (K, \phi)$ is called a $k$-critical bi-filtration,
if a) for all $\sigma \in K$, the cardinality of $\phi(\sigma)$ is at most $k$
and b) if $\sigma \subseteq \tau$, then $P_{\tau} \subseteq P_{\sigma}$.

In order to restrict a bi-filtration $F$ on a slice $\slice$,
we must, for each $\sigma \in K$, find the intersection of $\slice$ and the boundary of $P_{\sigma}$.
Obviously, the intersection point is the smallest push of the critical values of $\sigma$.
Thus the weighted restriction is determined by
\[
    \sigma \mapsto \wpush(\sigma, \slice) := \min_{i = 1,\dots,s} \wpush(p_i, \slice).
\]
As before, we write $\wpush(\sigma, \slice)$ or $\wpush_\sigma(\slice)$.
Similarly to \cref{sec:bound_primitive}, we now define variations
\[
    v(\sigma,\slicebox):=\max_{\slice\in\slicebox} |\weightedpush_\sigma(\slice)-\weightedpush_\sigma(\centerslice)|
\]
and
\[
v(F,\slicebox):=\max_{\sigma \in K} v(\sigma,\slicebox).
\]
\begin{mylem}
    \label{lem:min_max_multicrit}
    For $a_1, \dots, a_s$ and $b_1, \dots, b_s$ real,
    \[
        | \min_i(b_i) - \min_i(a_i) | \leq \max_i | b_i - a_i |.
    \]
\end{mylem}
\begin{proof}
    Let $c = \min(\min_i(a_i), \min_i(b_i))$. W.l.o.g., $c = a_1$.
    Then $ b_1 = a_1 + (b_1 - a_1) = c + | b_1 - a_1| \leq c + \max_i | b_i - a_i |$.
    Since $\min_i(b_i) \leq b_1$, this implies
    $\min_i(b_i) \leq c + \max_i | a_i - b_i |$. Hence
    $| \min_i(b_i) - \min_i(a_i) | = \min_i(b_i) - c \leq \max_i | a_i - b_i|$.
\end{proof}
\begin{mylem}
    The variation of a simplex $\sigma$ does not exceed the maximal variation
    of its critical values:
    $v(\sigma, B) \leq \max_{p_i \in \phi(\sigma)} v(p_i, B)$.
\end{mylem}
\begin{proof}
    In Lemma \ref{lem:min_max_multicrit}, set $a_i = \wpush(p_i, \slice)$ and 
    $b_i = \wpush(p_i, \centerslice)$.
    We obtain
    \[
    | \min_{i} \wpush(p_i, \slice) - \min_i \wpush(p_i, \centerslice) | \leq 
        \max_{i}  | \wpush(p_i, \slice) -  \wpush(p_i, \centerslice) |.
    \]
    Take maximum over $\slice$:
     \[
         \max_{\slice \in \slicebox} | \min_{i} \wpush(p_i, \slice) - \min_i \wpush(p_i, \centerslice) | \leq 
         \max_{\slice \in \slicebox} \max_{i}  | \wpush(p_i, \slice) -  \wpush(p_i, \centerslice) |.
    \]
    Interchange two $\max$ operations:
      \[
         \max_{\slice \in \slicebox} | \min_{i} \wpush(p_i, \slice) - \min_i \wpush(p_i, \centerslice) | \leq 
         \max_{i}  \max_{\slice \in \slicebox} | \wpush(p_i, \slice) -  \wpush(p_i, \centerslice) |.
    \]
    This is precisely the statement of the lemma.
\end{proof}
This lemma implies that our algorithm with the L-bound needs only one modification
for the $k$-critical case.  Namely, for each simplex $\sigma$ of $F^{1,2}$, 
we compute (using Theorem \ref{thm:corner-point}) the
maximal variation of the critical values of $\sigma$.
The lemma also implies that the $C$-bound and the $G$-bound remain unchanged.
\end{appendix}
\end{document}